\documentclass[lettersize,journal]{IEEEtran}
\usepackage{amsmath,amsfonts}
\usepackage{algorithmic}
\usepackage{algorithm}
\usepackage{array}
\usepackage[caption=false,font=normalsize,labelfont=sf,textfont=sf]{subfig}
\usepackage{textcomp}
\usepackage{stfloats}
\usepackage{url}
\usepackage{verbatim}
\usepackage{graphicx}
\usepackage{cite}
\usepackage{amsthm}
\usepackage{hyperref}
\usepackage{arydshln}
 \def\BibTeX{{\rm B\kern-.05em{\sc i\kern-.025em b}\kern-.08em
    T\kern-.1667em\lower.7ex\hbox{E}\kern-.125emX}}
\usepackage{balance}
\usepackage{xcolor}
\theoremstyle{plain}% Theorem-like structures provided by amsthm.sty
\newtheorem{theorem}{Theorem}
\newtheorem{lemma}{Lemma}
\newtheorem{corollary}{Corollary}
\newtheorem{proposition}{Proposition}
\newtheorem{problem}{Problem}

\theoremstyle{definition}
\newtheorem{definition}{Definition}
\newtheorem{example}{Example}

\theoremstyle{remark}
\newtheorem{remark}{Remark}

\newtheorem{assumption}{Assumption}

\begin{document}
\title{Data-Driven Cooperative Output Regulation of Continuous-Time 
 Multi-Agent Systems with Unknown Network Topology}
\author{Peng Ren, Yuqing Hao, Zhiyong Sun, Qingyun Wang, and Guanrong Chen, \IEEEmembership{Life~Fellow,~IEEE}
\thanks{The authors express gratitude to the National Nature Science Foundation of China (Grant Nos 12172020 and 11932003), Beijing National Science Foundation (Grant Nos 1222010), Young Elite Scientists Sponsorship Program by CAST (Grant Nos 2022QNRC001), the 111 Center under Grant B18002 and Beihang World TOP University Cooperation Program. (corresponding author: Y. Hao)}% <-this % stops a space
\thanks{P. Ren, Y. Hao, Q. Wang are with the Department of Dynamics and Control, Beihang University, Beijing, China (email: haoyq@buaa.edu.cn).}
\thanks{Z. Sun is with the College of Engineering, Peking University, Beijing,
China. (e-mail: zhiyong.sun@pku.edu.cn, sun.zhiyong.cn@gmail.com).}
\thanks{G. Chen is with the Department of Electrical Engineering, City
University of Hong Kong, Hong Kong SAR, China (e-mail: eegchen@cityu.edu.hk).}}

\markboth{Journal of \LaTeX\ Class Files,~Vol.~18, No.~9, September~2020}%
{How to Use the IEEEtran \LaTeX \ Templates}

\maketitle

\begin{abstract}
This paper investigates data-driven cooperative output regulation for continuous-time multi-agent systems with unknown network topology. %Unlike previous studies that rely on model identification, we employ a direct data-driven method to address the problem, reducing the computational burden.
Unlike existing studies that typically assume a known network topology to directly compute controller parameters, a novel approach is proposed that allows for the computation of the parameter without prior knowledge of the topology. A lower bound on the minimum non-zero eigenvalue of the Laplacian matrix is estimated using only edge weight bounds, enabling the output regulation controller design to be independent of global network information. Additionally, the common need for state derivative measurements is eliminated, reducing the amount of data requirements. Furthermore, necessary and sufficient conditions are established to ensure that the data are informative for cooperative output regulation, leading to the design of a distributed output regulation controller. For the case with noisy data, the bound of the output error is provided, which is positively correlated with the noise bound, and a distributed controller is constructed for the approximate cooperative output regulation. Finally, the effectiveness of the proposed methods is verified through numerical simulations.
\end{abstract}

\begin{IEEEkeywords}
Data-driven control, continuous-time multi-agent system, unknown network topology, cooperative output regulation, orthogonal polynomial basis. 
\end{IEEEkeywords}

\section{Introduction}
\IEEEPARstart{I}{n} recent years, cooperative control of multi-agent systems has been a focus of research, including synchronization control\cite{sun2021distributed,wen2020coordination,zhou2021terminal}, formation control\cite{sun2016exponential,xu2023adaptive}, cooperative output regulation\cite{cheng2019cooperative,liu2017cooperative}, and so on.    
The cooperative output regulation problem extends the traditional output regulation problem by requiring distributed control laws that ensure the endosystem outputs track the reference signals and reject disturbances from the exosystem. 
%It can also be viewed as a generalized leader-follower synchronization problem.
Initially, the problem was investigated in the context of static communication networks in \cite{su2011cooperative}. Subsequently, it was extended to switching communication networks in \cite{su2012cooperative1}. The study in \cite{huang2013cooperative1} addressed this problem using stabilizing $H_{\infty}$ controllers and internal models. An adaptive distributed observer was proposed in \cite{huang2016cooperative,liu2018adaptive} to solve the cooperative output regulation problem for discrete-time linear multi-agent systems. Event-triggering strategies were employed in \cite{hu2017cooperative} to address the cooperative output regulation issue. However, these existing approaches typically rely on accurate system models, which are challenging to obtain in practical scenarios.

A data-driven approach can solve control problems using measured data without precise model information. This method designs controllers based on data instead of precise model knowledge. Data-driven control methods are generally classified into two categories: indirect methods and direct methods. The indirect methods begin with model identification from data, followed by controller design and stability analysis based on the identified models  \cite{van2012subspace}. The direct methods, in contrast, design controllers and conduct theoretical analyses directly for unknown systems using collected data, without performing explicit model identification \cite{8933093,van2020data}. 

Data-driven approaches have been used to solve the output regulation problem. The data-driven output regulation problem for discrete-time systems was examined in \cite{trentelman2021informativity}. The robust data-driven output regulation control was studied in \cite{zhu2024data}. In many real-world applications, the single-system output regulation cannot address the challenges posed by networked systems. Therefore, investigating the data-driven cooperative output regulation problem is necessary to address these practical challenges effectively. The data-driven cooperative output regulation problem was studied using the indirect method in \cite{liang2024data}, which requires model identification using available data. In \cite{xie2023data}, the problem was addressed by adaptive dynamic programming. These studies primarily focus on discrete-time systems that cannot fully apply to physical systems like circuits, mechanical systems, and fluid dynamics, which are modeled as continuous-time systems. Therefore, it is necessary to investigate the data-driven cooperative output regulation problem for continuous-time systems. Additionally, these works focus primarily on unknown dynamics, assuming that the network topology is known. This assumption may not always hold in practical scenarios due to the uncertainties in the communication framework. Hence, an open problem remains to explore the cooperative control of multi-agent systems with an unknown network topology.

This paper addresses the data-driven cooperative output regulation problem for continuous-time systems with unknown network topology. The main contributions of this paper are three-fold:
\begin{itemize}
    \item The direct data-driven cooperative output regulation problem for continuous-time multi-agent systems is investigated. Compared with the recent works \cite{liang2024data} and \cite{xie2023data}, which focus on the data-driven cooperative output regulation problem of discrete-time systems, this study addresses continuous-time systems, thus filling a gap in the current studies of the concerned problem.
    % Additionally, we design a distributed controller solely based on exact datasets.
    %We introduce the data-driven transmission zeros as a sufficient condition and the data-driven output regulation equations as a necessary and sufficient condition to address the direct data-driven cooperative output regulation problem. Unlike \cite{liang2024data}, this paper eliminates the need for model identification. Consequently, our computation is reduced, and our methods are more concise. 
    \item The cooperative output regulation problem of multi-agent systems with unknown network topology is solved, complementing the studies of \cite{xie2023data} and \cite{jiao2021data}, which assume known topologies. The lower bound of the minimum non-zero eigenvalue of the Laplacian matrix is estimated solely based on the bounds of the available edge weights, thereby reducing the reliance on global information and avoiding the use of topology information.
    %Based on the bound, the distributed controller is designed.
    \item The data-driven approximate cooperative output regulation problem is examined in the presence of noisy data. By minimizing the norm of the error in the regulator equations, an approximate solution is derived, in which the convergence bound is positively correlated with the noise bound, facilitating approximate synchronization of all followers with leader's outputs. Additionally, a distributed controller is designed based on noisy data.
\end{itemize}

The rest of this paper is structured as follows. Section \uppercase\expandafter{\romannumeral2} introduces some preliminaries. The problem formulation is presented in Section \uppercase\expandafter{\romannumeral3}. The main results are derived and analyzed in Section \uppercase\expandafter{\romannumeral4}. Simulation examples are provided in Section \uppercase\expandafter{\romannumeral5}. Finally, Section \uppercase\expandafter{\romannumeral6} concludes the paper.

\section{PRELIMINARIES}

\subsection{Notation}
Let $\mathbb{R}$ denote the field of real numbers, and $\mathbb{N}$ denote the set of natural numbers. For a matrix $A$, $\sigma(A)$ represents the set of all eigenvalues of $A$, $\max \{\mathrm{Re}(A)\}$ denotes the real part of the largest eigenvalue
of $A$, $\min \{\mathrm{Re}(A)\}$ denotes the real part of the minimum eigenvalue
of $A$, $\|A\|$ denotes the 2-norm of $A$, and $ \mathrm{blockdiag}(A_{1},A_{2},\cdots,A_{n})$ denotes a partitioned diagonal matrix with diagonal matrices $A_{1},A_{2},\cdots,A_{n}$. Let $A\otimes B$ be the Kronecker product of matrices $A$ and $B$, $\mathrm{col}(x_{1},x_{2})$ be the column vector $[x_{1}^{T},x_{2}^{T}]^{T}$,  $\mathrm{Ker} A$ be the kernel of matrix $A$, $\mathrm{Im} A$ be the image or range of matrix $A$, and $A^{+}$ be the right inverse matrix of $A$, which satisfies $AA^{+}=I$. Matrices are assumed to be compatible with algebraic operations if their dimensions are not explicitly indicated. 

\subsection{Graph Theory}
A directed weighted graph is expressed as $\mathcal{G}=\{\mathcal{V}, \mathcal{E}, \mathcal{A} \}$, where $\mathcal{V}=\{\mathbf{v}_{0},\mathbf{v}_{1},...,\mathbf{v}_{N} \}$ is the set of nodes, $\mathcal{E}\subseteq \mathcal{V} \times \mathcal{V}$ is the set of directed edges $(\mathbf{v}_{i}, \mathbf{v}_{j})$, and $\mathcal{A}=\{a_{ij}\}$ is the adjacency matrix with nonnegative entries $a_{ij}$. The node $\mathbf{v}_{0}$ is designated as the leader, while the remaining nodes represent the followers. Define $a_{ij}>0$ if and only if $(\mathbf{v}_{j}, \mathbf{v}_{i}) \subseteq \mathcal{E}$, i.e., agent $i$ can receive information from agent $j$; otherwise, $a_{ij}=0$. A directed path from node $\mathbf{v}_0$ to node $\mathbf{v}_p$ consists of a sequence of edges $(\mathbf{v}_{k},\mathbf{v}_{k+1})$, $k=0,...,p-1$. If $(\mathbf{v}_{i}, \mathbf{v}_{j}) \in \mathcal{E}$ and $(\mathbf{v}_{j}, \mathbf{v}_{i}) \in \mathcal{E}$, with $a_{ij}=a_{ji}$ for all $(i,j)$, then the graph is termed a weighted undirected graph. A directed graph has a directed spanning tree if there exists a root node with directed paths to all other nodes without loops. Given a graph $\mathcal{G}$,  the degree matrix $\mathcal{D}$ is defined as
 $\mathcal{D}=\mathrm{diag}\{d_{1},...,d_{N}\}$, where each $d_{i}=\sum_{j=1}^{N} a_{ij}$. The Laplacian matrix $\mathcal{L}$ of $\mathcal{G}$ is given by $\mathcal{L}=\mathcal{D}-\mathcal{A}$, which can be expressed as
\begin{equation}
   \label{d0}
 \mathcal{L}= \left[
 \begin{array}{c|c}
    0 & 0_{1\times N}\\ \hline
    -[a_{10},...,a_{N0} ]^{T} & H
     \end{array}
    \right],
\end{equation}
where the matrix $H$ represents the graph among the followers.

\subsection{Orthogonal Polynomial Basis}

Let $\mathbb{I}=(t_{0},t_{1})$ denote an interval, where $t_{0}, t_{1} \in \mathbb{R}$. The space of square-integrable real-valued functions defined on $\mathbb{I}$ is denoted by $\mathcal{L}_{2}(\mathbb{I},\mathbb{R})$, equipped with the standard inner product $\langle f(t) , g(t) \rangle_{w}$, which is defined by $\langle f(t) , g(t) \rangle _{w} = \int_{\mathbb{I}}f(t)g(t)w(t)dt$ on $\mathcal{L}_{2}(\mathbb{I},\mathbb{R})$. The space of real square-summable sequences is denoted by $l_{2}(\mathbb{N},\mathbb{R})$. These notation and definitions naturally extend to vector-valued functions.

A basis $\{b_{k}\}$ is orthogonal if element $b_{i}$ and element $b_{j}$ are mutually orthogonal  with respect to the weight $w(t)$ on $\mathbb{I}$, i.e.,  $\langle b_{i}(t) , b_{j}(t) \rangle _{w} = \int_{\mathbb{I}}b_{i}(t)b_{j}(t)w(t)dt=0$. 
%A basis $\{b_{k}\}$ is complete if its linear span is dense in $ \mathcal{L}_{2}(\mathbb{I},\mathbb{R})$.
%Bessel's equality $\int_{\mathbb{I}} f(t)^{2} w(t) dt = \sum_{k=0}^{ \infty } \tilde{f}^{2}_{k}$ holds due to the complete orthonal basis. 
%Then, we can deduce $lim_{k \rightarrow \infty} \tilde{f}_{k} =0 $.
If each element $b_{k}$ of the basis is a polynomial, then $\{ b_{k}\}_{k \in \mathbb{N}}$ is called an Orthogonal Polynomial Base (OPB). For example, the Chebyshev polynomials on $\mathbb{I}=(-1,1)$ are defined as follows:
\begin{align*}
C_{0}(t)=&1, C_{1}(t)=t,\\
C_{n+1}(t)=&2tC_{n}(t)-C_{n-1}(t),
\end{align*}
where the weight function $w(t)=\frac{1}{\sqrt{1-t^{2}}}$.  Let $\tilde{f}_{k}$ denote the $k$-th coefficient of $f$ in the orthogonal basis representation. 
Define the bijective mapping between $\mathcal{L}_{2}(\mathbb{I},\mathbb{R})$ and $l _{2}(\mathbb{I},\mathbb{R})$ as follows:
 \begin{align*}
  \Omega : \mathcal{L}_{2}(\mathbb{I},\mathbb{R}) &\rightarrow  l_{2}(\mathbb{N},\mathbb{R}),\\
  f &\rightarrow \tilde{f}.
  \end{align*}
 The interval $\mathbb{I}$ will be transformed with $t \rightarrow \frac{2t-(t_{1}+t_{0})}{t_{1}-t_{0}}$ later when it is deeded. 
 %If function $f$ is Lipschitz continuous and the basis $\{ b_{k}  \}_{N \in \mathbb{N} }$ consists of the Chebyshev  polynomials, then the sequence $ \{ \sum^{N}_{k=0} \tilde{f}_{k} b_{k} \} $ converges absolutely and uniformly.
 Let $\{b_{k}\}_{k \in \mathbb{N}} $ be a set of OPBs. Then, $f$ can be expressed as
$$  f= \sum_{k=0}^{\infty} \tilde{f}_{k} \tilde{b}_{k}= \begin{bmatrix} \tilde{f}_{0}&\tilde{f}_{1}& \cdots\,  \end{bmatrix}  \begin{bmatrix}  b_{0}& b_{1}& \cdots\, \end{bmatrix}^{T}=\tilde{f}\mathrm{b},$$
where $\tilde{f}_{i}$ denotes the coefficient vector of $f$, $\tilde{f}=\begin{bmatrix} \tilde{f}_{0}&\tilde{f}_{1}& \cdots\, \end{bmatrix}$ is the coefficient matrix of $\{ b_{k}  \}_{N \in \mathbb{N} }$, and $ \mathrm{b}= \begin{bmatrix}  b_{0}& b_{1}& \cdots\, \end{bmatrix} ^{T}$ is the OPB matrix. 
%If $f$ is a vector function, it can be expressed as follows
%$$f= \begin{bmatrix} \tilde{f}_{1,0}&\tilde{f}_{1,1}& ...\\ \vdots &\vdots& ... \\ \tilde{f}_{n,0}&\tilde{f}_{n,1}& ... \end{bmatrix} \mathrm{b}=\tilde{f}\mathrm{b}.$$
The truncation of $f$ to degree $N$ is defined by
$$
 \Omega_{N}(f)=\sum_{k=0}^{N} \tilde{f}_{k} \tilde{b}_{k},
$$
and the error of truncation is 
$$
 \delta=f-\Omega_{N}(f)=\sum_{k=N+1}^{\infty} \tilde{f}_{k} \tilde{b}_{k}.
$$Suppose $f$ is differentiable. Then 
$$
 \frac{d}{dt}f=\sum_{k=0}^{\infty} \tilde{f}_{k} \frac{d}{dt}\tilde{b}_{k}
 =\sum_{k=0}^{\infty} \tilde{f}_{k} \sum_{j=0}^{\infty}d_{k,j}b_{k}.
$$
Let $\mathcal{D}_{\mathbin{b}}=[d_{k,j}]_{k,j \in \mathbb{N}}$. Then, $\frac{d}{dt}f$
can be formulated as follows:
$$ \frac{d}{dt}f=\tilde{f} \frac{d}{dt}\mathrm{b}=\tilde{f} \mathcal{D}_{\mathbin{b}} \mathrm{b}= \begin{bmatrix} \tilde{f}_{0}^{(1)}&\tilde{f}_{1}^{(1)}& ... \end{bmatrix}\mathrm{b} ,$$
where $\tilde{f}_{k}^{(1)}$ is the $k$-th coefficient vector of $\frac{d}{dt}f$ in the orthogonal basis representation. Partition $\mathcal{D}_{\mathbin{b}}$ as follows:
\begin{equation}
\label{d18}
 \mathcal{D}_{\mathbin{b}}= \begin{bmatrix} \mathcal{D}_{11}&\mathcal{D}_{12}\\
 \mathcal{D}_{21}&\mathcal{D}_{22}\end{bmatrix},
\end{equation}
where $\mathcal{D}_{11} \in \mathbb{R}^{(N+1) \times (N+1)}$ and  $\mathcal{D}_{22} \in \mathbb{R}^{\infty \times \infty}$. Let $\tilde{f}^{(1)}=\begin{bmatrix} \tilde{f}_{0}^{(1)}&\tilde{f}_{1}^{(1)}& ... \end{bmatrix}$ be the coefficient matrix. Then, the coefficient matrices $\tilde{f}$ and $\tilde{f} ^{(1)}$ can be partitioned as
\begin{equation}
    \label{d19}
 \tilde{f}= \begin{bmatrix} \Omega_{N}(f)& \tilde{f}^{'}
 \end{bmatrix},
 \tilde{f}^{(1)}= \begin{bmatrix} \Omega_{N}(\frac{d}{dt}f)& \tilde{f}^{(1)'}
 \end{bmatrix},
\end{equation}
where the coefficient matrices $\Omega_{N}(\frac{d}{dt}f)$ and $\frac{d}{dt}\Omega(f)$ are 
\begin{equation}
 \label{d20}
 \begin{aligned}
\Omega_{N} \left( \frac{d}{dt}f \right)=\Omega(f)\mathcal{D}_{11}+\tilde{f}^{'}\mathcal{D}_{21},
\frac{d}{dt}\Omega(f)=\Omega_{N}(f)\mathcal{D}_{11}.
\end{aligned}
\end{equation}

The following conditions determine whether the orders of differentiation and truncation are interchangeable in (\ref{d20}):
\begin{itemize}
  \item  If the differentiation and the truncation cannot commute, then $\Omega_{N}(\frac{d}{dt}f) \neq \frac{d}{dt}\Omega(f)$, which implies $\tilde{f}^{'}\mathcal{D}_{21} \neq 0$. 
    \item If the differentiation and the truncation can commute, then $\Omega_{N}(\frac{d}{dt}f) = \frac{d}{dt}\Omega(f)$, which implies $\tilde{f}^{'}\mathcal{D}_{21}=0$.
\end{itemize}

The following lemma shows that $\tilde{f}^{'}\mathcal{D}_{21}$ is bounded.

\begin{lemma} (see \label{l3} \cite{trefethen2019approximation})
Let $\{C_{k}\}_{k \in \mathbb{N}}$ be the Chebyshev basis, and $f \in \mathcal{L}_{2}(\mathbb{I},\mathbb{R})$ with $N \in \mathbb{N}$ and $N \geq 1$.
Suppose that $f$, $f^{(1)}$ are absolutely continuous, and that $V(f^{(2)})=\left \| \frac{d^{2}f}{d^{2}t} \right \|_{1}=\int_{-1}^{1}\left| \frac{d^{2}f(\tau)}{d^{2}t}\right|d\tau < \infty$. Partition $\mathcal{D}$ as in (\ref{d18}) and denote $\tilde{f}$, $\tilde{f}^{(1)}$ as in (\ref{d19}). Then, $\tilde{f^{'}}D_{21}$ satisfies
$\|\tilde{f^{'}}D_{21}\|_{2}\leq  \frac{2V(f^{(2)})}{\sqrt{\pi} (N-1)}$.
\end{lemma}

\section{PROBLEM FORMULATION}
Consider a leader-follower multi-agent system (MAS) composed of one leader and $N$ heterogeneous followers. The dynamics of follower $i \in \{1,...,N\}$ is
\begin{equation}
\label{d1.1}
\begin{aligned}
\dot{x}_{i}(t)&=\bar{A}_{i}x_{i}(t)+\bar{B}_{i}u_{i}(t)+\bar{E}_{i}d(t),\\
 y_{i}(t)&=\bar{C}_{i}x_{i}(t)+\bar{D}_{i}u_{i}(t)+\bar{F}_{i}d(t),
 \end{aligned}
\end{equation}
where $x_{i}(t)\in \mathbb{R}^{n_{i}}$ is the state, $u_{i}(t) \in \mathbb{R}^{m_{i}}$ is the input, $y_{i}(t)\in \mathbb{R}^{p}$ is the output, and $d(t) \in \mathbb{R}^{q_{1}}$ is the disturbance, which is generated by $\dot d(t)=A_{0d}d(t)$.
The dynamics of the leader is
\begin{equation}
\label{d2.1}
\begin{aligned}
\dot x_{0}(t)&=A_{0r}x_{0}(t),\\
 y_{0}(t)&=C_{0}x_{0}(t),\\
 \end{aligned}
\end{equation}
where $x_{0}(t) \in \mathbb{R}^{q_{2}}$ is the state of the leader, and $y_{0}(t) \in \mathbb{R}^{p}$ is the output of the leader.

Let $v(t)=\begin{bmatrix} x_{0}(t) \\ d(t)\end{bmatrix}$, $S= \begin{bmatrix} A_{0r} & 0 \\ 0 & A_{0d}\end{bmatrix}$, the exosystem can be expressed as
\begin{equation}
\label{d2}
 \dot{v}(t)=Sv(t),
\end{equation}
where $S\in \mathbb{R}^{q\times q}$ is a known matrix, and $v(t)\in \mathbb{R}^{q}$ is the tracking signals and/or disturbances.

The dynamics of the followers in (\ref{d1.1}) can be rewritten as
\begin{equation}
\label{d1}
\begin{aligned}
 \dot{x}_{i}(t)&=\bar{A}_{i}x_{i}(t)+\bar{B}_{i}u_{i}(t)+E_{i}v(t),\\
 e_{i}(t)&=y_{i}(t)-y_{0}(t)\\
 &=\bar{C}_{i}x_{i}(t)+\bar{D}_{i}u_{i}(t)+F_{i}v(t),
\end{aligned}
\end{equation}
where $e_{i}(t)\in \mathbb{R}^{p}$ is the tracking error, and $\begin{bmatrix} E_{i}\\ F_{i} \end{bmatrix}=\begin{bmatrix} 0& \bar{E}_{i}\\ -C_{0}& \bar{F}_{i}\end{bmatrix}$.

 The system matrices $\bar{A}_{i}\in \mathbb{R}^{n_{i}\times n_{i}}, \bar{B}_{i}\in \mathbb{R}^{n_{i}\times m_{i}}, \bar{C}_{i}\in \mathbb{R}^{p\times n_{i}}, \bar{D}_{i}\in \mathbb{R}^{p\times m_{i}}$ are unknown real matrices. The matrices $E_{i}\in \mathbb{R}^{n_{i}\times q}$ and $F_{i}\in \mathbb{R}^{p\times q}$ determine how disturbances and reference signals enter the system and are assumed to be known.

The following assumptions are made regarding the multi-agent systems.

\begin{assumption} \label{a1} The pairs $(\bar{A}_{i},\bar{B}_{i})$ are stabilizable. \end{assumption}

\begin{assumption} \label{a2} $S$ has no eigenvalues with negative real parts. \end{assumption}

\begin{assumption} \label{a3} The directed weighted graph $\mathcal{G}$ is unknown, but it contains a directed spanning tree with the leader node as the root.  \end{assumption}

\begin{assumption} \label{a4} There exist positive numbers $\varepsilon_{1}$ and $\varepsilon_{2}$, such that the nonzero elements of the matrix $\mathcal{L}$ satisfy  $\varepsilon_{1} \leq |\mathcal{L}_{ij} |\leq \varepsilon_{2} $ whenever $\mathcal{L}_{ij}\neq 0$. \end{assumption}

\begin{remark}
Assumptions \ref{a1}, \ref{a2} and \ref{a3} are common assumptions for the cooperative output regulation problem. Assumption~\ref{a4} ensures that the communication strength between different agents is bounded.
\end{remark}

Consider the following distributed controller:
\begin{equation}
\label{d3}
\begin{aligned}
 u_{i}(t)&=K_{1i}x_{i}(t)+K_{2i}\eta_{i}(t),\\
 \dot{\eta}_{i}(t)&=S\eta_{i}(t)+\mu\left[\sum_{j\in N_{i}} a_{ij}(\eta_{j}(t)-\eta_{i}(t))+a_{i0}(v(t)-\eta_{i}(t))\right],
\end{aligned}
\end{equation}
where $\eta_{i}(t)$ represents the state of the $i$-th controller, the gain matrices $K_{1i}$, $K_{2i}$, and the constant $\mu$ are to be designed.

Substituting Equation (\ref{d3}) into Equation~\eqref{d1}, the closed-loop system for the $i$-th agent is obtained as follows:
\begin{equation}
\label{d4}
\begin{aligned}
\dot{x}_{i}(t)&=(\bar{A}_{i}+\bar{B}_{i}K_{1i})x_{i}(t)+\bar{B}_{i}K_{2i}\eta_{i}(t)+E_{i}v(t),\\
\dot{\eta}_{i}(t)&=S\eta_{i}(t)+\mu\left[\sum_{j\in N_{i}} a_{ij}(\eta_{j}(t)-\eta_{i}(t))+a_{i0}(v(t)-\eta_{i}(t))\right],\\
e_{i}(t)&=(\bar{C}_{i}+\bar{D}_{i}K_{1i})x_{i}(t)+\bar{D}_{i}K_{2i}\eta_{i}(t)+F_{i}v(t).
\end{aligned}
\end{equation}

Let $\bar{A}=\mathrm{blockdiag}(\bar{A}_{1}, ...,\bar{A}_{N})$, $\bar{B}=\mathrm{blockdiag}(\bar{B}_{1}, ...,\bar{B}_{N})$, $\bar{C}=\mathrm{blockdiag}(\bar{C}{1}, ...,\bar{C}_{N})$, $\bar{D}=\mathrm{blockdiag}(\bar{D}_{1}, ...,\bar{D}_{N})$, $E=\mathrm{blockdiag}(E_{1}, ...,E_{N})$, $F=\mathrm{blockdiag}(F_{1}, ...,F_{N})$, $K_{1}=\mathrm{blockdiag}(K_{11}, ...,K_{1N})$, $K_{2}=\mathrm{blockdiag}(K_{21}, ...,K_{2N})$, $\hat{v}(t)=1_{N} \otimes v(t)$, $\eta(t)=\mathrm{col}(\eta_{1}(t), ...,\eta_{N}(t))$, $x(t)=\mathrm{col}(x_{1}(t), ...,x_{N}(t))$, $e(t)=\mathrm{col}(e_{1}(t), ...,e_{N}(t))$. Then, the compact form of Equation~(\ref{d4}) can be expressed as follows:
\begin{equation}
\label{d5}
\begin{aligned}
\dot{x}(t)&=(\bar{A}+\bar{B}K_{1})x(t)+\bar{B}K_{2}\eta(t)+E\hat{v}(t),\\
\dot{\eta}(t)&=[(I_{N}\otimes S)-\mu (H\otimes I_{q})]\eta(t)+\mu(H\otimes I_{q})\hat{v}(t),\\
e(t)&=(\bar{C}+\bar{D}K_{1})x(t)+\bar{D}K_{2}\eta(t)+F\hat{v}(t).
\end{aligned}
\end{equation}

Let $x_{c}(t)=\mathrm{col}(x(t), \eta(t))$. Equation (\ref{d5}) can be transformed to the following form:
\begin{equation}
\label{d6}
\begin{aligned}
\dot{x}_{c}(t)=&\begin{bmatrix} \bar{A}+\bar{B}K_{1} & \bar{B}K_{2}\\ 0 & (I_{N}\otimes S)-\mu (H\otimes I_{q}) \end{bmatrix}x_{c}(t) \\ &+\begin{bmatrix} E \\ \mu(H\otimes I_{q}) \end{bmatrix} \tilde{v}(t),\\
e(t)=&\begin{bmatrix} \bar{C}+\bar{D}K_{1} & \bar{D}K_{2}  \end{bmatrix}x_{c}(t)+F\tilde{v}(t).
\end{aligned}
\end{equation}

Let $\bar{A}_{C}=\begin{bmatrix} \bar{A}+\bar{B}K_{1} & \bar{B}K_{2}\\ 0 & (I_{N}\otimes S)-\mu (H\otimes I_{q}) \end{bmatrix}$, $B_{C}=\begin{bmatrix} E \\ \mu(H\otimes I_{q}) \end{bmatrix}$, $\bar{C}_{C}=\begin{bmatrix} \bar{C}+\bar{D}K_{1} & \bar{D}K_{2}  \end{bmatrix}$, and $D_{C}=F$. Equation~(\ref{d6}) can be compactly expressed as follows:
\begin{align*}
\dot{x}_{c}(t)&=\bar{A}_{C}x_{c}(t)+B_{C}\tilde{v}(t),\\
e(t)&=\bar{C}_{C}x_{c}(t)+D_{C}\tilde{v}(t).
\end{align*}

The cooperative output regulation is defined as follows.
\begin{definition} (see \cite{su2011cooperative})
Consider the systems given by \eqref{d2} and \eqref{d1}, along with the graph $\mathcal{G}$. The cooperative output regulation involves designing a distributed controller such that the following conditions are satisfied:
\begin{itemize}
    \item The matrix $\bar A_{C}$ is Hurwitz.
    \item The tracking error $e(t)$ satisfies $\lim \limits_{t \rightarrow + \infty} e(t)=0$.
\end{itemize}
\end{definition}

The following lemma solves the cooperative output regulation problem when the matrices $A_{i},B_{i},C_{i},D_{i}$ are known.

\begin{lemma} \label{l6} 
Under Assumptions \ref{a1}-\ref{a4}, suppose that $\mu$ is sufficiently large. The cooperative output regulation problem is solvable using the controller (\ref{d3}), if either of the following two conditions is satisfied:
\begin{itemize}
\item 
The transmission zero condition is satisfied, which means
\begin{equation} \label{d28.1}
\mathrm{rank}\begin{bmatrix}
   \bar A_{i}-\lambda_{S} I_{i} & \bar B_{i} \\
   \bar C_{i} & \bar D_{i}
\end{bmatrix}=n_{i}+p_{i},
\end{equation}
for all $\lambda_{S}\in\sigma(S)$, where $\sigma(S)$ denotes the spectrum of $S$ and  $i \in \{1,...,N\}$.
\item The following regulator equations have a solution $(\Pi_{i},\Gamma_{i})$:
\begin{equation}
\label{d34.1}
\begin{aligned}
  \Pi_{i} S&=\bar A_{i} \Pi_{i}+ \bar B_{i} \Gamma_{i}+E_{i},\\
  0&=\bar C_{i}\Pi_{i}+\bar D_{i} \Gamma_{i}+F_{i},\\
\end{aligned}
\end{equation}
where $\Gamma_{i}=K_{1i} \Pi_{i}+K_{2i}$.
\end{itemize}
\end{lemma}

Since the matrices $\bar{A}_{i}$, $\bar{B}_{i}$, $\bar{C}_{i}$, and $\bar{D}_{i}$ are unknown,  the following lemma is introduced to represent the continuous-time multi-agent systems.

\begin{lemma} \label{l4} (see \cite{rapisarda2023orthogonal})
Let $ \{ b_{k}\}_{k \in \mathbb{N}}$  be the complete OPBs for $\mathcal{L}_{2}(\mathbb{I},\mathbb{R})$. The following statements are equivalent:
\begin{itemize}
    \item The input $u$ and the state $x$ satisfy the continuous-time system $\dot{x}=Ax+Bu$.
    \item The sequences $\tilde{x}=\{\tilde{x}\}_{k \in \mathbb{N}}$, $\tilde{x}^{(1)}=\{\tilde{x}^{(1)}\}_{k \in \mathbb{N}}$, and $\tilde{u}=\{\tilde{u}\}_{k \in \mathbb{N}},$ which correspond to the OPB representations, satisfy
$$
\tilde{x}^{(1)}\mathrm{b}=\tilde{x} \mathcal{D}_{\mathbin{b}}\mathrm{b}=A\tilde{x}\mathrm{b}+B\tilde{u}\mathrm{b},
$$
i.e.,
$$\tilde{x}^{(1)}=\tilde{x} \mathcal{D}_{\mathbin{b}}=A\tilde{x}+B\tilde{u}.$$
\end{itemize}
\end{lemma}

\begin{remark}
The OPB method does not require measuring the derivative of the state, as it can be computed from the state. Consequently, this method reduces the amount of data that are needed to be measured.
\end{remark}

According to Lemma \ref{l4} and Equation \eqref{d20}, Equation~(\ref{d1}) can be expressed as follows:
\begin{subequations}
\label{d9}
\begin{align}
X_{i}\mathcal{D}_{11}&=A_{i}X_{i}+B_{i}U_{i}+E_{i}V+W_{i},\\
 \Im_{i}&=C_{i}X_{i}+D_{i}U_{i}+F_{i}V,
\end{align}
\end{subequations}
 where the coefficient vectors $\tilde{x}_{ik}$, $\tilde{u}_{ik}$,  $\tilde{e}_{ik}$, $\tilde{v}_{k}$ corresponding to the OPB $\{b_{k}\}_{k \in \mathbb{N}}$ for $\mathcal{L}_{2}(\mathbb{I},\mathbb{R})$ are collected from the following matrices: 
 \begin{equation}
\label{d8}
\begin{aligned}
X_{i}&= \begin{bmatrix} \tilde{x}_{i0} & \tilde{x}_{i1} & \tilde{x}_{i2} & ... & \tilde{x}_{iN}
\end{bmatrix} \in \mathbb{R}^{n_{i}\times (N+1)},\\
U_{i}&= \begin{bmatrix} \tilde{u}_{i0} & \tilde{u}_{i1} & \tilde{u}_{i2} & ... & \tilde{u}_{iN}
\end{bmatrix}\in \mathbb{R}^{m_{i}\times (N+1)},\\
\Im_{i}&= \begin{bmatrix} \tilde{e}_{i0} & \tilde{e}_{i1} & \tilde{e}_{i2} & ... & \tilde{e}_{iN}
\end{bmatrix}\in \mathbb{R}^{p\times (N+1)},\\
V&= \begin{bmatrix} \tilde{v}_{0} & \tilde{v}_{1} & \tilde{v}_{2} & ... & \tilde{v}_{N}
\end{bmatrix} \in \mathbb{R}^{q\times (N+1)}.\\
 \end{aligned}
\end{equation}
The noise term $W_{i}$ is defined as $W_{i} = -\tilde{X_{i}'}D_{21}$, where $\tilde{X_{i}'}$ represents the truncation error of $\tilde{x_{i}}$. The matrix $\tilde{x}_{i}$ is partitioned analogously to Equation \eqref{d19}, as $\tilde{x}_{i} = \begin{bmatrix} X_{i} & \tilde{X_{i}'} \end{bmatrix}$. According to Lemma \ref{l3}, the 2-norm of $W_{i}$ is bounded by $\|W_{i}\|_{2}\leq \frac{2V(x_{i}^{(2)})}{\sqrt{\pi}. (N-1)}$. Consequently, the noise $W_{i}$ satisfies $W_{i}W_{i}^{T}\leq c_{i}I_{i}$, where $c_{i}$ is a constant defined as $c_{i} = \left ( \frac{2V(x_{i}^{(2)})}{\sqrt{\pi} (N-1)}\right )^{2}$. Equivalently, the following inequality holds:
 \begin{equation}
\label{d8.8}
\begin{aligned}
 \begin{bmatrix} I_{i} & W_{i} \end{bmatrix} \begin{bmatrix} c_{i}I_{i} & 0 \\ 0 & -I_{i}\end{bmatrix}  \begin{bmatrix} I_{i} \\ W_{i}^{T} \end{bmatrix} \geq 0.
 \end{aligned}
\end{equation}

 A model that can generate the measured data (\ref{d8}) is formulated as follows:
 \begin{equation}
\label{d10}
\begin{aligned}
{\sum}_{Fi}=\left\{\begin{bmatrix} A_{i} & B_{i} \\ C_{i} & D_{i} \end{bmatrix}\Bigg|\begin{bmatrix} A_{i} & B_{i} \\ C_{i} & D_{i} \end{bmatrix}\begin{bmatrix} X_{i}\\ U_{i} \end{bmatrix}=\begin{bmatrix} X_{i}\mathcal{D}_{11}-E_{i}V-W_{i}\\ \Im_{i}-F_{i}V \end{bmatrix}\right\}.
 \end{aligned}
\end{equation}
Obviously, $\begin{bmatrix} \bar{A}_{i} & \bar{B}_{i} \\ \bar{C}_{i} & \bar{D}_{i} \end{bmatrix} \in \sum_{Fi}$. 
Moreover, let
\begin{equation}
\label{d22}
{\sum}_{Fi}^{0}~=~\left\{\begin{bmatrix} A_{i} & B_{i} \\ C_{i} & D_{i} \end{bmatrix}~\Bigg|~ \begin{bmatrix} A_{i} & B_{i} \\ C_{i} & D_{i} \end{bmatrix}\begin{bmatrix} X_{i}\\ U_{i} \end{bmatrix}=0 \right\}.
\end{equation}

%Before formulating the data-driven cooperative output regulation problem, informativity is defined as follows.

%\begin{definition} \label{de2} 
%\cite{van2020data} The data matrices are informative for property $\nabla$ if ${\sum}_{F} \subseteq \sum_{\nabla}$, where $\sum_{\nabla}$ denotes a system set in which all systems have property $\nabla$, and ${\sum}_{F}$ is a model set that can generate the measured data %\eqref{d8}.
%\end{definition}

%\begin{assumption}
%\label{a5} 
%The data $(X_{i},U_{i})$ are informative for stabilization.
%\end{assumption}

When the data are exact, i.e., $W_{i}=0$, 
%meaning that the approximation error of the derivative is negligible, i.e.,  $X_{i}\mathcal{D}_{21} = 0$, 
the data-driven cooperative output regulation problem is formulated as follows.

\begin{problem} \label{p1}
Consider the multi-agent system (\ref{d2}) and (\ref{d1}). For any $\begin{bmatrix} A_{i} & B_{i} \\ C_{i} & D_{i} \end{bmatrix} \in \sum_{Fi}$ and the unknown graph $\mathcal{G}$, find conditions and design the controller for the structure of (\ref{d3}) based on the exact data
%which the data are informative for the solvability of 
such that the cooperative output regulation problem is solvable, i.e., both the following objectives hold:
\begin{itemize}
\item The matrix
 \begin{equation}
\label{d23}
\begin{aligned}
A_{C}=\begin{bmatrix} A+BK_{1} & BK_{2}\\ 0 & (I_{N}\otimes S)-\mu (H\otimes I_{q}) \end{bmatrix}
\end{aligned}
\end{equation}
is stabilizable, where $A=\mathrm{blockdiag}(A_{1}, ...,A_{N})$, $B=\mathrm{blockdiag}(B_{1}, ...,B_{N})$.
\item The tracking error $e(t)$ satisfies
    \begin{equation}
        \label{d24}
    \lim _{t \rightarrow +\infty}  e(t)=\lim _{t \rightarrow +\infty}  (C_{C}x_{c}(t)+D_{C}\tilde{v}(t))=0,
    \end{equation}
where $C_{C}=\begin{bmatrix} C+DK_{1} & DK_{2} \end{bmatrix}$, $C=\mathrm{blockdiag}(C_{1}, ...,C_{N})$, and $D=\mathrm{blockdiag}(D_{1}, ...,D_{N})$.
\end{itemize}
\end{problem}

When the data are noisy, i.e., $W_{i} \neq 0$,
%which means that the approximation error of the derivative cannot be neglected, i.e.,  $X_{i}\mathcal{D}_{21} \neq 0$, 
the data-driven approximate cooperative output regulation problem is formulated as follows.

\begin{problem}
Consider the multi-agent system (\ref{d2}) and (\ref{d1}). For any $\begin{bmatrix} A_{i} & B_{i} \\ C_{i} & D_{i} \end{bmatrix} \in \sum_{Fi}$ with an unknown graph $\mathcal{G}$,
%under which the data are informative for the solvability of the cooperative 
find conditions and design the controller for the structure of (\ref{d3}) using the noisy data such that the approximate cooperative output regulation problem is solvable, i.e., both the following objectives hold:
\begin{itemize}
\item The matrix (\ref{d23}) is quadratically stabilizable.
\item The tracking error $e(t)$ is ultimately uniformly bounded, i.e., 
    \begin{equation}
        \label{d25}
    \lim _{t \rightarrow +\infty} \|e(t)\|=\lim _{t \rightarrow +\infty} \|C_{c}x_{c}(t)+D_{c}\tilde{v}(t)\| \leq \gamma,
    \end{equation}
    where $\gamma$ is a constant.
\end{itemize}
\end{problem}

\section{MAIN RESULT}
 This section investigates the cooperative output regulation problem under the data-informativity framework. The main results are divided into two subsections. The first subsection addresses the data-driven cooperative output regulation problem based on exact data. The second subsection addresses the approximate data-driven cooperative output regulation problem with noise data.
\subsection{Exact data case}
When $W_{i}=0$, the set of systems can be  characterized as follows:
 \begin{equation}
\label{d27}
\begin{aligned}
{\sum}_{Fi}^{e}=\left\{\begin{bmatrix} A_{i} & B_{i} \\ C_{i} & D_{i} \end{bmatrix}\Bigg|\begin{bmatrix} A_{i} & B_{i} \\ C_{i} & D_{i} \end{bmatrix}\begin{bmatrix} X_{i}\\ U_{i} \end{bmatrix}=\begin{bmatrix} X_{i}\mathcal{D}_{11}-E_{i}V\\ \Im_{i}-F_{i}V \end{bmatrix}\right\}.
 \end{aligned}
\end{equation}
The controller (\ref{d3}) designed based on the available exact data is expressed as follows:
\begin{equation}
\label{d3.1}
\begin{aligned}
 U_{i}&=K_{1i}X_{i}+K_{2i}\bar{\eta}_{i},\\
 \bar{\eta}_{i}\mathcal{D}_{11}&=S\bar{\eta}_{i}+\mu(\sum_{j\in N_{i}} a_{ij}(\bar{\eta}_{j}-\bar{\eta}_{i})+a_{i0}(V-\bar{\eta}_{i})),
\end{aligned}
\end{equation}
where $\bar{\eta}_{i}= \begin{bmatrix} \tilde{\eta}_{i0} & \tilde{\eta}_{i1} & \tilde{\eta}_{i2} & ... & \tilde{\eta}_{iN} 
\end{bmatrix} $ and $\tilde{\eta}_{ik}$ is the coefficient vector of $ b_{k}$.
Substituting Equation (\ref{d3.1}) into Equation~(\ref{d27}), the augmented system is obtained, as 
%\begin{equation}
%\label{d3.2}
%\begin{aligned}
%X\mathcal{D}_{11}&=(A+BK_{1})X+BK_{2}\bar{\eta}+E\bar{V},\\
%\bar{\eta}{D}_{11}&=((I_{N}\otimes S)-\mu (H\otimes I_{q}))\bar{\eta}+\mu(H\otimes %I_{q})\bar{V},\\
%\Im&=(C+DK_{1})X+DK_{2}\bar{\eta}+F\bar{V},
%\end{aligned}
%\end{equation}
\begin{equation}
\label{d3.3}
\begin{aligned}
X_{c}\mathcal{D}_{11}=&\begin{bmatrix} A+BK_{1} & BK_{2}\\ 0 & (I_{N}\otimes S)-\mu (H\otimes I_{q}) \end{bmatrix}X_{c} \\ &+\begin{bmatrix} E \\ \mu(H\otimes I_{q}) \end{bmatrix} \bar{V},\\
\Im=&\begin{bmatrix} C+DK_{1} & DK_{2}  \end{bmatrix}X_{c}+F\bar{V},
\end{aligned}
\end{equation}
 where $X_{c}=\mathrm{col}(X, \bar{\eta})$, $X=\mathrm{col}(X_{1}, ...,X_{N})$, $\bar{\eta}=\mathrm{col}(\bar{\eta}_{1}, ...,\bar{\eta}_{N})$, $\Im=\mathrm{col}(\Im_{1}, ...,\Im_{N})$, and $\bar{V}=1_{N} \otimes V$. The system \eqref{d3.3} can be described as
\begin{equation}
\label{d7.1}
\begin{aligned}
{X}_{c}\mathcal{D}_{11}&=A_{C}X_{c}+B_{C}\bar{V},\\
\Im&={C}_{C}X_{c}+D_{C}\bar{V},
\end{aligned}
\end{equation}
where $A_{C}$, $B_{C}$, $C_{C}$ and $D_{C}$ have been defined previously.

The cooperative output regulation problem involves two primary objectives:  ensuring that $A_C$ is stabilizable, and achieving $\lim \limits_{t \rightarrow + \infty} e(t)=0$. $A_C$ is  stabilizable if both $A+BK_1$ and $(I_{N}\otimes S)-\mu (H\otimes I_{q}) $ are stabilizable. Based on these criteria, the following subsection is organized into three parts: 1) The condition for stabilizability of $A_{C}$ is investigated and the matrix $K_{1i}$ is computed; 2) The unknown network topology is investigated, yielding an estimation of $\mu$. 3) The condition of $\lim \limits_{t \rightarrow + \infty} e(t)=0$ is studied, with $K_{2i}$ designed accordingly.

\subsubsection{The stabilizability of $A_C$} \hfill 

The following definition relates to the informativity for the stabilizability of $A_C$.

\begin{definition} \label{de2}
The data $(X_{i}, U_{i})$ are informative for the stabilizability of $A_{C}$ if $\sum_{D}\subseteq \sum_{A_{C}} $, where $\sum_{D}=\left\{(A, B) | XD_{11}=AX+BU+E\bar V\right\}$ and $\sum_{A_{C}}$=\{ $(A, B)$ $|$ There exist $K_{1}$, $K_{2}$ and $\mu$ such that $A_{C}$ in \eqref{d23} is Hurwitz\}.
\end{definition}

 The following lemma is introduced to compute $K_{1i}$ such that the data $(X_{i}, U_{i})$ are informative for the stabilizability of $A_{C}$. 

\begin{lemma} \label{l5} 
\cite{van2020data} Suppose the data $(X_i, U_i)$ are informative for stabilizability, and let $K_{1i}$ be a feedback gain such that $\sum_{F_i}^{e1} \subseteq \sum_{K_{1i}}$, where $\sum_{K_{1i}} = \{(A_i, B_i) \mid A_i + B_i K_{1i} \text{ is stable}\}$, $\sum_{F_i}^{e1}=\{(A_i, B_i) \mid X_i\mathcal{D}_{11}=A_iX_i + B_iU_i+E_iV \}$ and $i \in \{1,...,N\}$.
Then
\begin{equation}
\label{d27.7}
\begin{aligned}
\mathrm{Im} \begin{bmatrix} I_{i} \\ K_{1i} \end{bmatrix} \subseteq \mathrm{Im} \begin{bmatrix} X_{i} \\ U_{i} \end{bmatrix}.
\end{aligned}
\end{equation}
\end{lemma}

Lemma \ref{l5} provides a condition for calculating $K_{1i}$. 
The following theorem establishes the necessary and sufficient condition such that the data $(X_{i}, U_{i})$ are informative for the stabilizability of $A_{C}$. 

\begin{theorem}
\label{t0} Under Assumptions \ref{a2}-\ref{a4}, assume that $\mu$ is sufficiently large. The data $(X_{i}, U_{i})$ are informative for the stabilizability of $A_{c}$, if and only if, there exists a right inverse $X_{i}^{+}$ of $X_{i}$ such that $(X_{i}\mathcal{D}_{11}-E_{i}V)X_{i}^{+}$ is stable for all  $i \in \{1,...,N\}$.
\end{theorem}

\begin{proof}
Sufficiency: First, it is needed to investigate the eigenvalues of $I_{N}\otimes S-\mu (H\otimes I_{q})$ based on \eqref{d23}. Define nonsingular matrices $Z_{1}$ and $Z_{2}$, such that the following equations hold:
$$
 Z_{1}^{-1}SZ_{1}=\tilde{S}=\begin{bmatrix} \lambda_{S_{1}} & \varpi_{1}& 0 & 0 &...& 0  \\ 0  & \lambda_{S_{2}} & \varpi_{2} & 0 & ... & 0\\ \vdots & \vdots & \ddots & \ddots &...& 0\\ 0 & 0 & ... & 0 & \lambda_{S_{q-1}} & \varpi_{q-1} \\ 0 & 0& 0 & ... & 0 &\lambda_{S_{q}} \end{bmatrix},$$\\
$$ Z_{2}^{-1}HZ_{2}=\tilde{H}=\begin{bmatrix} \lambda_{H_{1}} & \bar \varpi_{1}& 0 & 0 &...& 0  \\ 0  & \lambda_{H_{2}} & \bar \varpi_{2} & 0 & ... & 0\\ \vdots & \vdots & \ddots & \ddots &...& 0\\ 0 & 0 & ... & 0 &  \lambda_{H_{N-1}} & \bar \varpi_{N-1} \\ 0 & 0& 0 & ... & 0 & \lambda_{H_{N}} \end{bmatrix}, 
$$
 where $\bar \varpi_{i} \in \{0,1\}$ and $\varpi_{j} \in \{0,1\}$.
 It follows that
\begin{equation} \label{d31.1}
\begin{aligned}
 &(Z_{2}^{-1} \otimes Z_{1}^{-1})[I_{N}\otimes S-\mu (H\otimes I_{q})](Z_{2} \otimes Z_{1}) \\
 =&[I_{N}\otimes (Z_{1}^{-1}SZ_{1})]-[\mu (Z_{2}^{-1}HZ_{2}^{-1}) \otimes I_{q}]\\
 =&(I_{N}\otimes \tilde{S})-\mu (\tilde{H} \otimes I_{q}).
 \end{aligned}
 \end{equation} 
Equation (\ref{d31.1}) implies that the eigenvalues of $(I_{N}\otimes S)-\mu (H\otimes I_{q})$ are $\lambda_{S_{i}}-\mu \lambda_{H_{j}}$, where $i=1,...,q$, $j=1,...,N$. Since $\mu$ is sufficiently large, with $\mu>\frac{\max\{\mathrm{Re}(\lambda_{S})\}}{\min\{\mathrm{Re}(\lambda_{H})\}}$, the matrix $I_{N}\otimes S-\mu (H\otimes I_{q})$ is Hurwitz.

According to the set (\ref{d27}), one has $X_{i}\mathcal{D}_{11}-E_{i}V=A_{i}X_{i}+B_{i}U_{i}$. Since $(X_{i}\mathcal{D}_{11}-E_{i}V)X_{i}^{+}$ is stable, it follows that
$$(X_{i}\mathcal{D}_{11}-E_{i}V)X_{i}^{+}=(A_{i}X_{i}+B_{i}U_{i})X_{i}^{+}=A_{i}+B_{i}K_{1i},$$ 
where $K_{1i}=U_{i}X_{i}^{+}$.
Therefore, $(A_{i},B_{i}) \in {\sum}_{Fi}^{e}$ is stabilizable. Hence, the data $(X_{i}, U_{i})$ are informative for the stabilizability of $A_{C}$.

Necessity: Since the data $(X_{i}, U_{i})$ are informative for the stabilizability of $A_{c}$, there exists a constant $\mu$ such that $I_{N}\otimes S-\mu (H\otimes I_{q})$ is Hurwitz and $A_{i}+B_{i}K_{1i}$ is stable for all $(A_{i},B_{i}) \in {\sum}_{Fi}^{e}$. By Equation (\ref{d27.7}), it can be deduced that $X_{i}$ has full row-rank and that there exists a right inverse $X_{i}^{+}$ of $X_{i}$, such that
\begin{equation}
\label{d119}
\begin{aligned}
\begin{bmatrix} I_{i} \\ K_{1i} \end{bmatrix} =  \begin{bmatrix} X_{i} \\ U_{i} \end{bmatrix} X_{i}^{+}.
\end{aligned}
\end{equation}
Then, the following equation can be obtained: 
$$A_{i}+B_{i}K_{1i}=\begin{bmatrix} A_{i} & B_{i} \end{bmatrix}\begin{bmatrix} I_{i} \\ K_{1i} \end{bmatrix}=\begin{bmatrix} A_{i} & B_{i} \end{bmatrix}\begin{bmatrix} X_{i} \\ U_{i} \end{bmatrix} X_{i}^{+}.$$ By Equations (\ref{d9}) and (\ref{d27}), $A_{i}+B_{i}K_{1i}=(X_{i}\mathcal{D}_{11}-E_{i}V) X_{i}^{+}$. Therefore, $(X_{i}\mathcal{D}_{11}-E_{i}V) X_{i}^{+}$ is stable.
\end{proof}

Theorem \ref{t0} provides the necessary and sufficient condition for $A_C$ to be Hurwitz.
The following theorem presents a method to compute $K_{1i}$.
\begin{theorem}\label{t2}
The data $(X_{i}, U_{i})$ are informative for stabilizability, if and only if, there exists a matrix $\theta_{i} \in \mathbb{R}^{(N+1) \times n_{i} }$ satisfying
\begin{equation}
\begin{aligned}
\label{d17.7}
(X_{i}\theta_{i})&=(X_{i}\theta_{i})^{T} > 0,\\
\theta_{i}^{T}(X_{i}\mathcal{D}_{11} -E_{i}V)^{T}&+(X_{i}\mathcal{D}_{11} -E_{i}V)\theta_{i} < 0,
\end{aligned}
\end{equation}
for $i \in \{1,...,N\}$. Furthermore, $K_{1i}$ can be computed by
\begin{equation}
\label{d28}
\begin{aligned}
K_{1i}=U_{i}X_{i}^{+}=U_{i}\theta_{i}(X_{i}\theta_{i})^{-1}.
\end{aligned}
\end{equation}
\end{theorem}

\begin{proof}
Sufficiency: Since the matrix $X_{i}\theta_{i}$ is symmetric and positive definite, the matrix $X_{i}$ has full rank.  It follows that $X_{i}^{+}=\theta_{i}(X_{i}\theta_{i})^{-1}$. By left-multiplying $(X_{i}\theta_{i})^{-T}$ and right-multiplying $(X_{i}\theta_{i})^{-1}$ to the second formula of Equation $(\ref{d17.7})$, one obtains
$(X_{i}\theta_{i})^{-T}\theta_{i}^{T}(X_{i}\mathcal{D}_{11} -E_{i}V)^{T}(X_{i}\theta_{i})^{-1}+(X_{i}\theta_{i})^{-T}(X_{i}\mathcal{D}_{11} -E_{i}V)\theta_{i}(X_{i}\theta_{i})^{-1}.$ Combing it with Equation \eqref{d28}, one has that $(A_{i}+B_{i}K_{1i})^{T}(X_{i}\theta_{i})^{-1}+(X_{i}\theta_{i})^{-T}(A_{i}+B_{i}K_{1i})< 0$. Define $P_{i}=(X_{i}\theta_{i})^{-T}$. Then, the Lyapunov function $V_{i}=x_{i}^{T}P_{i}x_{i}$ satisfies 
\begin{equation}
\label{d28.012}
(A_{i}+B_{i}K_{1i})^{T}P_{i}+P_{i}(A_{i}+B_{i}K_{1i})<0.
\end{equation}
Therefore, the data $(X_{i}, U_{i})$ are informative for stabilizability. 

Necessity:
Since the data $(X_{i}, U_{i})$ are informative for stabilizability, there exists a positive definite matrix $P_{i}=P_{i}^{T}$, such that the Lyapunov function $V_{i}$ satisfies \eqref{d28.012} for all $(A_i,B_i) \in \sum_{Fi}^{e1}$. Then, it can be concluded that $$P_{i}[(X_{i}\mathcal{D}_{11}-E_{i}V) X_{i}^{+}]^{T}+[(X_{i}\mathcal{D}_{11}-E_{i}V) X_{i}^{+}]P_{i}<0.$$  Define $\theta_{i}=X_{i}^{+}P_{i}$. Equation \eqref{d17.7} can be obtained. It is obvious that $X_{i}^{+}=\theta_{i}(X_{i}\theta_{i})^{-1}$. Combing Equation \eqref{d119}, Equation \eqref{d28} can be obtained. 
\end{proof}

\subsubsection{Unkown network topology} \hfill 

This part estimates a lower bound of the minimum non-zero eigenvalue of the Laplacian matrix based solely on edge weight bounds and provides the corresponding lower bound for $\mu$. The following theorem proposes a sufficient condition on $\mu$ such that  $(I_{N}\otimes S)-\mu (H\otimes I_{q})$ is stabilizable. 

\begin{theorem}
\label{t1} Under Assumptions \ref{a2}-\ref{a4}, the matrix $(I_{N}\otimes S)-\mu (H\otimes I_{q})$ is stabilizable, if 
\begin{equation}
\label{d31}
\mu > \frac{\max\{\mathrm{Re}(\lambda_{S})\}N(2\varepsilon_{2})^{N-1}}{\varepsilon_{1}^{N}},
\end{equation}
where $\max\{\mathrm{Re}(\lambda_{S})\}$ is the largest real part of the eigenvalues of $S$. 
\end{theorem}

\begin{proof}
By Equation (\ref{d31.1}), it can be deduced that the eigenvalues of $(I_{N}\otimes S)-\mu (H\otimes I_{q})$ are $\lambda_{S_{i}}-\mu \lambda_{H_{j}}$, where $i=1,...,q$, $j=1,...,N$.  According to \cite{shivakumar1996two}, the following inequality holds:
$$
\min \{\mathrm{Re}(\lambda_{H})\}=\frac{1}{\max\{\mathrm{Re}(\lambda_{H^{-1}})\}} \geq \frac{1}{\|H^{-1}\|_{\infty}}.
$$
Thus, an upper bound for $\|H^{-1}\|_{\infty}$ needs to be found. Since $H^{-1}=\frac{H^{*}}{\det(H)}$, where $H^{*}$ is the adjoint matrix of $H$, it follows that
$$
\|H^{-1}\|_{\infty}=\frac{\|H^{*}\|_{\infty}}{|\det(H)|}=\frac{\|H^{*}\|_{\infty}}{\det(H)}.
$$
According to Equation (\ref{d0}), the diagonal element $H_{ii}$ of the matrix $H$ satisfies the following equality:
\begin{equation}
\label{952}
H_{ii}=\sum_{i\neq j,i=1}^{N}|H_{ij}|+a_{i0},
\end{equation}
where $a_{i0}$ is the weight of the edge from the leader to follower $i$. 
Under Assumption \ref{a4}, the nonzero elements $\mathcal{L}_{ij}$ in the matrix $\mathcal{L}$  satisfy the inequality $\varepsilon_{1} \leq |\mathcal{L}_{ij} |\leq \varepsilon_{2}.$ As a result, $\varepsilon_{1} \leq |H_{ij} |\leq \varepsilon_{2}$. Combing Gerschgorin’s Circle Theorem and Equation (\ref{952}), the following inequalities hold:
\begin{align*}
\max \{\lambda_{H}\}
&\leq H_{ii}+\sum_{i\neq j,i=1}^{N}|H_{ij}|
\leq H_{ii}+\sum_{i\neq j,i=1}^{N}|H_{ij}|+a_{i0}\\
&=2H_{ii} \leq 2\varepsilon_{2}.
\end{align*}
%\begin{figure}[ht]
%  \centering
%  \includegraphics[width=160pt]{gaier.pdf}
 % \caption{An illustrative demonstration of Gerschgorin's circles.}
 % \label{fig5}
%\end{figure}
Thus, the elements $H^{*}_{ij}$ of the matrix $H^{*}$ satisfy the following inequality:
$$H^{*}_{ij}=(-1)^{i+j}M_{ji}\leq\lambda_{M_{ji_{1}}}...\lambda_{M_{ji_{N-1}}}\leq (2\varepsilon_{2})^{N-1},$$
where $M_{ji}$ is the cofactor of the element $H_{ij}$ and $\lambda_{M_{ji_{k}}}$ is the eigenvalue of $M_{ij}$.
Then, it follows that
\begin{equation}
\label{d2345}
\|H^{*}\|_{\infty}=\max_{i} \left\{\sum_{j=1}^{N}|H_{ij}| \right\} \leq N(2\varepsilon_{2})^{N-1}. 
\end{equation}
A lower bound of $\det(H)$ will be derived. 
%It is evident that  $\det(H_{1\times 1}) \geq \varepsilon_{1}$ when $N=1$. 
When $N=2$, the matrix $H$ can be written as 
$$H_{2\times 2}=\begin{bmatrix}
a_{12}+a_{10} & -a_{12}\\
-a_{21} &a_{21}+a_{20}\\
\end{bmatrix},$$ 
    where $a_{12}+a_{21}>0$, $a_{10}+a_{20}>0$ or $a_{12}=a_{21}=0$, $a_{10},a_{20}>0$, because the matrix contains a directed spanning tree. There are $3$ spanning trees: $a_{21},a_{10}\neq 0$, otherwise $0$; $a_{20},a_{12}\neq 0$, otherwise $0$; $a_{20},a_{10}\neq 0$, otherwise $0$. Therefore, the lower bound of $\det(H_{2\times 2})$  can be computed as follows:
\begin{align*}
\det(H_{2\times 2})&=a_{21}a_{10}+a_{20}a_{12}+a_{20}a_{10}\geq \varepsilon_{1}^{2}.
\end{align*}
 According to Cayley's formula \cite{west2001introduction}, the number of spanning trees on the labeled nodes is at most $(k+1)^{k-1}$ when $N=k$. By the matrix tree theorem \cite{west2001introduction}, $\det(H)$ is the number of spanning trees if $a_{ij}=1$ or $0$  for all $i\neq j$. It can be deduced that $\det(H)$ is the sum of the products of  $k$ weights of $(k+1)^{k-1}$spanning trees, i.e.,
$$\det(H)=\sum_{(k+1)^{k-1}}a_{1i_{1}}a_{2i_{2}}...a_{ki_{k}},$$
where $a_{1i_{1}}a_{2i_{2}}...a_{ki_{k}}$ is a spanning tree.
By Assumption \ref{a3}, the graph contains at least one directed spanning tree, i.e.,
$$\det(H)\geq a_{1i_{1}}a_{2i_{2}}...a_{ki_{k}}.$$
Besides, a lower bound of the nonzero elements $a_{ij}$ is $\varepsilon_{1}$, i.e.,  $a_{ij}\geq \varepsilon_{1}$ when $a_{ij}\neq 0$.
Therefore, the lower bound of $\det(H)$ can be estimated
\begin{equation}
\label{d2333}
\det(H)\geq \varepsilon_{1}^{N}.
\end{equation}
Combing equations (\ref{d2345}) and (\ref{d2333}), the upper bound of $\|H^{-1}\|_{\infty}$ can be derived: 
$$\|H^{-1}\|_{\infty}\leq \frac{N(2\varepsilon_{2})^{N-1}}{\varepsilon_{1}^{N}}.$$
Then, it follows that
\begin{equation}\label{d31.222}
\min\{\mathrm{Re}(\lambda_{H})\} \geq \frac{1}{\|H^{-1}\|_{\infty}}  \geq  \frac{\varepsilon_{1}^{N}}{N(2\varepsilon_{2})^{N-1}}.
\end{equation}
Hence, the inequality (\ref{d31}) can be deduced.  Given that $\max\{\mathrm{Re}(\lambda_{S})\} -\mu \min\{\mathrm{Re}(\lambda_{H})\}<0$ according to Inequality (\ref{d31.222}), it follows that $(I_{N}\otimes S)-\mu (H\otimes I_{q})$ is stable. 
\end{proof}

\begin{remark}
Different from \cite{xie2023data} and \cite{jiao2021data}, which consider a known network topology, Theorem \ref{t1} estimates a lower bound on the minimum non-zero eigenvalue of the
Laplacian matrix, ensuring that the matrix $(I_{N}\otimes S) - \mu (H \otimes I_{q})$ is Hurwitz.
\end{remark}

\begin{remark}
In \cite{ li2016distributed}, a method is proposed to calculate the upper bound of $\|H^{-1}\|_{\infty} $, i.e.,
$
\|H^{-1}\|_{\infty} \leq \pi_{N},
$
where $\pi_{N}=A^{N-2}_{N-1}(\frac{4\max(a)}{\min(a)})^{N-2}(\frac{2\min(a)+\max(a)}{(\min(a))^{2}})+\frac{2}{\min(a)}\sum^{N-3}_{k=0}A_{N-1}^{k}(\frac{4\max(a)}{\min(a)})^{k}$, $\min(a) =\min\{a_{ij}:(\mathbf{v}_{j} ,\mathbf{v}_{i}) \in \mathcal{E}\}$, $\max(a) =\max\{a_{ij}:(\mathbf{v}_{j} ,\mathbf{v}_{i}) \in \mathcal{E}\}$ and $A_{n}^{m}=\frac{n!}{(n-m)!}$. Combining Assumption \ref{a4}, this condition can be extended to the case where the Laplacian matrix is unknown, i.e., $\varpi_{N} \geq \pi_{N}$,
where $$\varpi_{N}=A^{N-2}_{N-1}\left(\frac{4\varepsilon_{2}}{\varepsilon_{1}}\right)^{N-2}\left(\frac{2\varepsilon_{1}+\varepsilon_{2}}{\varepsilon_{1}^{2}}\right)+\frac{2}{\varepsilon_{1}}\sum^{N-3}_{k=0}A_{N-1}^{k}\left(\frac{4\varepsilon_{2}}{\varepsilon_{1}}\right)^{k}.$$ 
Then, it follows that
\begin{equation}\label{d31.333}
\min \{\mathrm{Re}(\lambda_{H})\} \geq \frac{1}{\|H^{-1}\|_{\infty}}  \geq  \frac{1}{\varpi_{N}}.  
\end{equation}
\end{remark}

\begin{example}
Let $H=\begin{bmatrix}
12 &0& 0& -7\\ -5 & 10 &0& -5\\ 0 & -5 &10 &-5\\-5& 0 &0 &5
\end{bmatrix}$, it can be verified that $\varepsilon_{1}=\min(a)=5$, $\max(a)=7$, and $\varepsilon_{2}=12$. Invoking Inequalities (\ref{d31.222}) and $(\ref{d31.333})$, the lower bounds of $\mathrm{Re}(\lambda_{H})$ are found as shown in Table \ref{tab2}. It is evident that Inequality (\ref{d31.222}) offers a better estimate compared to other bounds.
Inequality (\ref{d31.333}) is conservative compared to $\frac{1}{\pi_{N}}$ due to the fact that less information is available about the graph.
\begin{table}
\begin{center}
\caption{Estimates of the Lower bound of $\lambda_{H}$}
\label{tab2}
\begin{tabular}{| c | c | c | c |}
\hline
$\min{\mathrm{Re}(\lambda_{H})}$ & $\frac{\varepsilon_{1}^{N}}{N(2\varepsilon_{2})^{N-1}}$ & $\frac{1}{\pi_{N}}$ & $\frac{1}{\varpi_{N}}$ \\
\hline
 $1.6261$& $ 0.0113$
& $0.0074$& $ 0.0020 $\\
\hline
\end{tabular}
\end{center}
\end{table}
\end{example}

\begin{remark}
When the connections between the leader and the followers is unknown, but the connections among the followers are known, with all $a_{ij}=1$ for $i=0,\ldots, N$, more precise lower bounds for $\lambda_{H}$ and $\mu$ can be estimated. Refer to Fig. \ref{fig6} for a visual representation and Table \ref{tab1} for detailed results, where  $\delta_{N}=3\delta_{N-1}-\delta_{N-2}$, $\delta_{1}=1$ and $\delta_{2}=3$. All results in Table \ref{tab1} are derived using an iterative computing method. When the graph is complete, $\det(H)\geq N^{N-2}$ according to the matrix tree theorem \cite{west2001introduction}. The equation $\det(H)= N^{N-2}$ holds when the leader $0$ connects to a single follower $i$. The norm $\|H^{*}\|_{\infty}$ satisfies $\|H^{*}\|_{\infty}\leq (N+1)^{N-1}$. When the leader $0$ connects to all followers, $\|H^{*}\|_{\infty}= (N+1)^{N-1}$. Therefore, $\mathrm{Re}(\lambda_{H})$ satisfies $\mathrm{Re}(\lambda_{H})\geq \frac{\varepsilon_{1}^{N}}{N(2\varepsilon_{2})^{N-1}}$. Besides, the more information available about the graph, the more accurate the estimations of the lower bounds for both $\mathrm{Re}(\lambda_{H})$ and $\mu$.
\begin{figure}[ht]
  \centering
  \includegraphics[width=170pt]{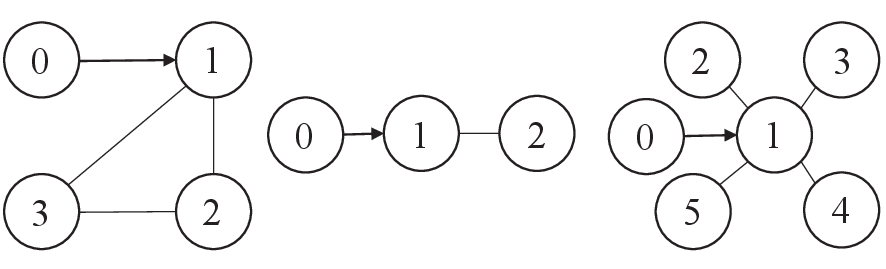}
  \caption{Followers are connected on different graphs.}
  \label{fig6}
\end{figure}

\begin{table}
\begin{center}
\caption{Estimates of the Lower bound of $\lambda_{H}$ and $\mu$ in different graphs}
\label{tab1}
\begin{tabular} {| c | c | c | }
\hline
Graph  & Lower bound of  &
Lower bound of \\
 & $ \mathrm{Re}( \lambda_{H})$ & $\mu$\\
\hline
Complete
graph
& $\frac{N^{N-2}}{(N+1)^{N-1}} $&$\frac{\max\{\mathrm{Re}(\lambda_{S})\}{(N+1)^{N-1}}}{N^{N-2}}$\\ 
\hline
Undirected path& $\frac{1}{\delta_{N}}$ & $\max\{\mathrm{Re}(\lambda_{S})\}\delta_{N}$\\
\hline
Star
& $\frac{1}{(N+1)2^{N-2}}$ & $\max \{\mathrm{Re}(\lambda_{S})\}(N+1)2^{N-2} $ \\
\hline
\end{tabular}
\end{center}
\end{table}
\end{remark}

\begin{remark}
If the leader is directly connected to all followers while the connections among the followers are unknown, $H$ is a strictly diagonally dominant matrix (see Fig. \ref{fig7}).
\begin{figure}[ht]
  \centering
  \includegraphics[width=80pt]{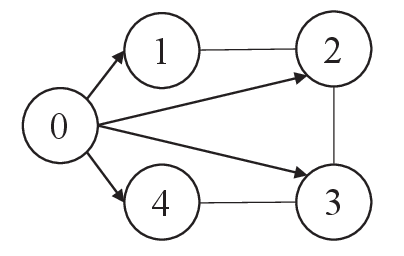}
  \caption{Leader connects to all followers and followers are connected by an undirected graph.}
  \label{fig7}
\end{figure} 
By Gerschgorin’s Circle Theorem,
% the following set is defined:
%$$
%\vartheta_{i}=\{\chi_{i} \in \mathbb{C} : |\chi_{i}- (\lambda_{S_{i}}-\mu \lambda_{H_{i}})| \leq %|\varpi_{i} - \mu \bar \varpi_{i} |\}.
%$$
 the following inequality holds:
$$\mu H_{ii}-S_{ii}> \mu\sum_{i\neq j}|H_{ij}|+\sum_{i\neq j}|S_{ij}|.$$
Then,
$$\mu>\frac{S_{ii}+\sum_{i\neq j}|S_{ij}|}{ H_{ii}-\sum_{i\neq j}|H_{ij}|}.$$ Note that $ H_{ii}-\sum_{i\neq j}|H_{ij}|=a_{i0}$, where $a_{i0}$ denotes the in-degree from the leader to the followers. Assuming $a_{i0}>\bar \varepsilon$, it can be deduced that
\begin{equation} \label{d31.444}
\mu>\frac{\max(S_{ii}+\sum_{i\neq j}|S_{ij}|)}{\bar \varepsilon}.
\end{equation}

%Furthermore, if all $a_{i0}=1$ and $N$ followers are connected by an undirected graph, 
%According to the definition of directed graph algebraic connectivity in \cite{wu2005algebraic}, the following equation holds
%$$\min \lambda_{H}=&\min_{x_{N+1}\bot 1_{N+1}} \frac{x_{N+1}^{T}\mathcal{L}x_{N+1}}%{x_{N+1}^{T}x_{N+1}}.$$
%Then, the following inequality can be established 
% \begin{align*}
%\min \lambda_{H}=&\min_{ x_{N+1}\bot 1_{N+1}} \frac{x_{N}^{T}(\mathcal{M}+I) x_{N}+x_{1}^{2}}%{x_{N+1}^{T}x_{N+1}} \\
%\geq & \min_{ x_{N+1}\bot 1_{N+1}} \frac{x_{N}^{T}x_{N}+x_{1}^{2}}{x_{N+1}^{T}x_{N+1}}\\
%\geq &\min_{ x_{N} \bot 1_{N}} \frac{x_{N+1}^{T}x_{N+1}}{x_{N+1}^{T}x_{N+1}}
%=1,
%\end{align*}
%and the lower bound of $\mu$ can be written as follows
%\begin{equation} \label{1.111}
%\mu>\max\{\mathrm{Re}(\lambda_{S})\}.
%\end{equation}
\end{remark}

If $S$ is unknown, then the following proposition can be proved to show that the matrix $(I_{N}\otimes S)-\mu (H\otimes I_{q})$ is Hurwitz.
\begin{proposition}
\label{co1}
Under Assumptions \ref{a2}-\ref{a4}, assume that $S$ is unknown, and all elements $s_{ij}$ of $S$ satisfy  $|S_{ij}| \leq \epsilon$. The matrix $(I_{N}\otimes S)-\mu (H\otimes I_{q})$ is Hurwitz, if the constant $\mu$ satisfies 
\begin{equation}
\label{d32}
\mu > \frac {q\epsilon N\epsilon(2\varepsilon_{2})^{N-1}}{\varepsilon_{1}^{N}}.
\end{equation}
\end{proposition}

\begin{proof}
According to Gerschgorin’s Circle Theorem, the following inequality holds: 
\begin{equation}
\label{d32.1}
\begin{aligned}
\mathrm{Re}(\lambda_{S})&\leq S_{ii}+\sum^{q}_{j=1,i\neq j} |S_{ij}|  
\leq \|S\|_{\infty}.
\end{aligned}
\end{equation}
    Since $|S_{ij}|\leq \epsilon$, Inequality (\ref{d32.1}) gives rise to
\begin{equation}
\label{d322.1}
\begin{aligned}
\mathrm{Re}(\lambda_{S})\leq \|S\|_{\infty}\leq  q\epsilon.
\end{aligned}
\end{equation}
By applying Equations (\ref{d31.222}) and (\ref{d322.1}), one obtains Equation (\ref{d32}). Consequently, it follows that $\mathrm{Re}(\lambda_{S_{i}}) - \mu \mathrm{Re}(\lambda_{H_{j}}) < 0$. This implies that $(I_{N} \otimes S) - \mu (H \otimes I_{q})$ is stable. 
\end{proof}

%\begin{example}
%Assume the real matrices $\mathcal{L}$ and $S$ are 
%$$\mathcal{L}=\begin{bmatrix}
%0&0&0&0\\
%-3& 3 &0 &0\\
%-3 & -3 &6& 0 \\ 
% -3& 0& -3 & 6
%\end{bmatrix},
%S=\begin{bmatrix} 3& 1& 2\\ 2& 5& 2\\ 0& 2& 3
%\end{bmatrix}.$$
%It is obvious that $\varepsilon_{1}=3$, $\varepsilon_{2}=6$, $\min(\lambda_{H})=3$, $\max\{\mathrm{Re}(\lambda_{S})\}=7$, $\bar \varepsilon=3$. Invoking Equations (\ref{d31.222}), (\ref{d31.444}) and (\ref{d32}), the lower bounds of $\mu$ can be computed, see also Table \ref{tab3}. Obviously, the more information known about the graph, the more accurate the bound on $\mu$ becomes.
%\begin{table}
%\begin{center}
%\caption{Estimates the Lower bound of $\lambda_{H}$}
%\label{tab3}
%\begin{tabular}{| c | c | c | c |}
%\hline
%$\min{\mu}$ & (\ref{d31}) & (\ref{d31.444}) & (\ref{d32})\\
%\hline
% $2.333$& $ 112$
%& $3$  & $240$\\
%\hline
%\end{tabular}
%\end{center}
%\end{table}
%\end{example}

%Theorem \ref{t1} and Proposition \ref{co1} provide methods for $(I_{N}\otimes S)-\mu (H\otimes I_{q})$ to be Hurwitz. 

\subsubsection{The cooperative output regulation} \hfill 
%Theorem \ref{t0} provides the necessary and sufficient condition such that the data $(X_{i}, U_{i})$ are informative for stabilizability of $A_{C}$. 

The following definition is introduced to formalize the informativity for the cooperative output regulation problem based on Lemma \ref{l6}.

\begin{definition} \label{de3}
Under Assumptions \ref{a2}-\ref{a4}, suppose $\mu$ is sufficiently large. The data $(X_{i}, U_{i},\Im_{i})$ are informative for the cooperative output regulation problem for $i \in \{1,...,N\}$,
if  $\sum_{D}\subseteq \sum_{A_{C}} $, and either of the following two conditions is satisfied:
\begin{itemize}
    \item 
$\sum_{Fi}^{e}\subseteq \sum_{zero}^{e}$,
where \\$\sum_{zero}=\left\{  \begin{bmatrix}
   A_{i} &  B_{i} \\
    C_{i} &  D_{i}
\end{bmatrix} \Bigg | \mathrm{rank}\begin{bmatrix}
   A_{i}-\lambda_{S} I_{i} &  B_{i} \\
    C_{i} &  D_{i}
\end{bmatrix}=n_{i}+p_{i}
\right\}. $
\item 
$\sum_{Fi} \subseteq \sum_{regu}$, where $\sum_{regu}=\left\{\begin{bmatrix}
    A_{i}&B_{i}\\C_{i}&D_{i}
\end{bmatrix} \bigg| \textit{Equation \eqref{d34.1} has a soluntion $(\Pi_{i},\Gamma_{i})$}  \right \}.$
\end{itemize}
\end{definition}

%\begin{remark}
%The rank condition (\ref{d28.1}) is a sufficient condition for the solvability of the linear output regulation problem in \cite{huang2004nonlinear}. The solvability of the output regulation problem still holds when the pair $\mathrm{vec} \left( \begin{bmatrix} E_{i} \\ F_{i} \end{bmatrix}\right)  \in \mathrm{Im} \left( \begin{bmatrix} I_{n_{i}} & 0 \\ 0 & 0_{p_{i} \times m_{i}}  \end{bmatrix} - I_{q} \otimes \begin{bmatrix} A_{i} & B_{i}\\ C_{i} & D_{i} \end{bmatrix} \right) $.
%\end{remark}

\begin{remark}
Definition \ref{de3} encompasses Definition \ref{de2}, and consists of two key components: $A_{C}$ being Hurwitz and $\lim \limits_{t \rightarrow + \infty} e(t)=0$. Besides, the first condition in Definition \ref{de3} is sufficient, while the second condition is both necessary and sufficient, as shown in \cite{su2011cooperative}.
\end{remark}

The following theorem presents a sufficient condition for data-driven cooperative output regulation problem.

\begin{theorem}
\label{t3}
Under Assumptions \ref{a2}-\ref{a4}, suppose $\mu$ is sufficiently  large. The data $(X_{i}, U_{i},\Im_{i})$ are informative for the cooperative output regulation problem, if there exists a right inverse $X_{i}^{+}$ of $X_{i}$ such that $(X_{i}\mathcal{D}_{11}-E_{i}V)X_{i}^{+}$ is stable, and the following data-driven transmission zero condition holds:
\begin{equation}
\label{d29}
\mathrm{rank}\begin{bmatrix} X_{i}\mathcal{D}_{11}-E_{i}V-\lambda_{S} X_{i} \\ \Im_{i}-F_{i}V \end{bmatrix}=n_{i}+p_{i},
\end{equation}
 for all $\lambda_{S} \in\sigma(S)$, and $i \in \{1,...,N\}$.
\end{theorem}

\begin{proof}
%Sufficiency:
Let $\Delta_{i} =\begin{bmatrix} A_{i}-\lambda_{S} I_{i} & B_{i} \\ C_{i} & D_{i} \end{bmatrix}$, $\Psi_{i}=\begin{bmatrix} X_{i} \\ U_{i} \end{bmatrix}$, $\Theta_{i}=\begin{bmatrix}  X_{i}\mathcal{D}_{11}-E_{i}V-\lambda_{S} X_{i} \\ \Im_{i}-F_{i}V \end{bmatrix}$.
Combining Equations (\ref{d27}) and (\ref{d29}), the following equation holds: 
\begin{align*}
  \begin{bmatrix}  X_{i}\mathcal{D}_{11}-E_{i}V-\lambda_{S} X_{i} \\ \Im_{i}-F_{i}V \end{bmatrix}=\begin{bmatrix} A_{i}-\lambda_{S} I_{i} & B_{i} \\ C_{i} & D_{i} \end{bmatrix}\begin{bmatrix} X_{i} \\ U_{i} \end{bmatrix},
\end{align*}
  i.e., $\Theta_{i}=\Delta_{i}\Psi_{i}$. It is evident that $n_{i}+p_{i}=\mathrm{rank}(\Delta_{i}\Psi_{i}) \leq \mathrm{rank}(\Delta_{i}) \leq n_{i}+p_{i}$. Therefore, $\mathrm{rank}(\Delta_{i})=n_{i}+p_{i}$. Furthermore, since there exists a right inverse $X_{i}^{+}$ of $X_{i}$ such that $(X_{i}\mathcal{D}_{11}-E_{i}V)X_{i}^{+}$ is stable, the data $(X_{i}, U_{i})$ are informative for the stabilizability of $A_{C}$. Consequently, the data $(X_{i}, U_{i},\Im_{i})$ are informative for the cooperative output regulation problem. 
\end{proof}

The following theorem establishes the necessary and sufficient condition for the data-driven cooperative output regulation problem. Additionally, it provides a method to compute the matrix $K_{2i}$.

\begin{theorem}
\label{t4}
Under Assumptions \ref{a2}-\ref{a4}, suppose $\mu$ is sufficiently large. The data $(X_{i}, U_{i},\Im_{i})$ are informative for the cooperative output regulation problem, if and only if, there exists a right inverse $X_{i}^{+}$ of $X_{i}$ such that $(X_{i}\mathcal{D}_{\mathbin{b}}-E_{i}V)X_{i}^{+}$ is stable, and the following data-driven regulator equations have a solution $M_{i}$:
\begin{equation}
\label{d33}
\begin{aligned}
(X_{i}\mathcal{D}_{11}-E_{i}V)M_{i}&=X_{i}M_{i}S-E_{i}, \\
(\Im_{i}-F_{i}V)M_{i}&=-F_{i},
\end{aligned}
\end{equation}
for $i \in \{1,...,N\}$. Besides, $K_{2i}$ can be computed as 
\begin{equation} \label{d34} K_{2i}=(U_{i}-K_{1i}X_{i})M_{i} .\end{equation}

\end{theorem}

\begin{proof}
Sufficiency: 
By Equation (\ref{d27}), $X_{i}\mathcal{D}_{11}-E_{i}V=A_{i}X_{i}+B_{i}U_{i}$, $\Im_{i}-F_{i}V=C_{i}X_{i}+D_{i}U_{i}$. Substituting these equations into Equation (\ref{d33}), the following equations can be obtained:
\begin{equation}\label{d35}
\begin{aligned}
A_{i}X_{i}M_{i}+B_{i}U_{i}M_{i}+E_{i}&=X_{i}M_{i}S, \\
C_{i}X_{i}M_{i}+D_{i}U_{i}M_{i}+F_{i}&=0.
\end{aligned}
\end{equation}
 Let $X_{i}M_{i}=\Pi_{i}$, $U_{i}M_{i}=\Gamma_{i}$. By using Equation (\ref{d35}), Equation (\ref{d34.1}) can be obtained, where $\Gamma_{i}=K_{1i} \Pi_{i}+K_{2i}$. Besides, since there exists a right inverse $X_{i}^{+}$ of $X_{i}$ such that $(X_{i}\mathcal{D}_{\mathbin{b}}-E_{i}V)X_{i}^{+}$ is stable, the data $(X_{i}, U_{i})$ are informative for the stabilizability of $A_{C}$. 
 
 Necessity: 
 Equation (\ref{d34.1}) can be expressed as $\begin{bmatrix} A_{i} & B_{i} \\ C_{i} & D_{i} \end{bmatrix}\begin{bmatrix} \Pi_{i} \\ \Gamma_{i} \end{bmatrix} = \begin{bmatrix} \Pi_{i} S-E_{i} \\ -F_{i} \end{bmatrix} $. The homogeneous equation is $\begin{bmatrix} A_{i} & B_{i} \\ C_{i} & D_{i} \end{bmatrix}\begin{bmatrix} \Pi_{i} \\ \Gamma_{i} \end{bmatrix} =\begin{bmatrix} 0 \\ 0 \end{bmatrix}.$ Moreover, $\begin{bmatrix} \Pi_{i}^{T} & \Gamma_{i}^{T} \end{bmatrix}\begin{bmatrix} A_{i} & B_{i} \\ C_{i} & D_{i} \end{bmatrix}^{T}=\begin{bmatrix} 0 & 0 \end{bmatrix}$, i.e., $ \begin{bmatrix} A_{i} & B_{i} \\ C_{i} & D_{i} \end{bmatrix} \in {\sum}_{Fi}^{0} $ with $$ \begin{bmatrix} A_{i} & B_{i} \\ C_{i} & D_{i} \end{bmatrix}^{T} \in \mathrm{Ker} \begin{bmatrix} \Pi_{i}^{T} & \Gamma_{i}^{T} \end{bmatrix}.$$ Since Equation (\ref{d22}) holds, it follows that $ \begin{bmatrix} X_{i}^{T} & U_{i}^{T} \end{bmatrix}\begin{bmatrix} A_{i} & B_{i} \\ C_{i} & D_{i} \end{bmatrix}^{T}=\begin{bmatrix} 0 & 0 \end{bmatrix}$, i.e.,  
$$\begin{bmatrix} A_{i} & B_{i} \\ C_{i} & D_{i} \end{bmatrix}^{T} \in  \mathrm{Ker} \begin{bmatrix} X_{i}^{T} & U_{i}^{T} \end{bmatrix}.$$
 Since the data $(X_{i}, U_{i}, \Im_{i})$ are informative for the cooperative output regulation problem, it follows that $  \mathrm{Ker}  \begin{bmatrix} X_{i}^{T} & U_{i}^{T}\end{bmatrix} \subseteq  \mathrm{Ker} \begin{bmatrix} \Pi_{i}^{T} & \Gamma_{i}^{T} \end{bmatrix} $. This implies that $  \mathrm{Im} \begin{bmatrix} \Pi_{i} \\ \Gamma_{i} \end{bmatrix} \subseteq  \mathrm{Im} \begin{bmatrix} X_{i} \\ U_{i}\end{bmatrix}$, which can be expressed as $\begin{bmatrix} \Pi_{i} \\ \Gamma_{i} \end{bmatrix} = \begin{bmatrix} X_{i} \\ U_{i}\end{bmatrix} M_{i}.$ Recalling that $\Gamma_{i}=K_{1i} \Pi_{i}+K_{2i}$, the following equation is obtained:
 \begin{equation} \label{d22.2}
\begin{bmatrix} \Pi_{i} \\ \Gamma_{i} \end{bmatrix}=\begin{bmatrix} \Pi_{i} \\ K_{1i}\Pi_{i}+K_{2i}\Gamma_{i} \end{bmatrix}=\begin{bmatrix} X_{i} \\ U_{i} \end{bmatrix}M_{i}.
\end{equation}
By combining Equations (\ref{d27}), (\ref{d34.1}) and (\ref{d22.2}),  Equation (\ref{d33}) can be obtained, where $M_{i}$ is the solution of Equation (\ref{d33}). Therefore, Equation (\ref{d33}) has a solution $M_{i}$, and Equation (\ref{d34}) holds. Besides, since the data $(X_{i}, U_{i})$ are informative for the stabilizability of $A_{C}$, there exists a right inverse $X_{i}^{+}$ of $X_{i}$ such that $(X_{i}\mathcal{D}_{\mathbin{b}}-E_{i}V)X_{i}^{+}$ is stable.
\end{proof}

\begin{remark}
The cooperative output regulation problem for discrete-time systems under the persistency of excitation condition has been explored in \cite{liang2024data,xie2023data}. This paper focuses on the cooperative output regulation problem for continuous-time systems, without the need for the persistency of excitation condition as required in \cite{liang2024data}. This approach reduces the data requirements.
\end{remark}

%The first condition of Problem \ref{p1} has been addressed.
Algorithm \ref{suanfa} is derived from the main conclusions of this paper, detailing the steps for implementing the data-driven cooperative output regulation to compute the controller (\ref{d3}).
\begin{algorithm}
%\floatname{algorithm}{Algorithm}%¸ü¸ÄËã·¨Ç°×ºÃû³Æ
%\renewcommand{\algorithmicrequire}{\textbf{Input:}}%¸ü¸ÄÊäÈëÃû³Æ
%\renewcommand{\algorithmicensure}{\textbf{Output:}}%¸ü¸ÄÊä³öÃû³Æ
\footnotesize
\caption{Data-Driven Cooperative Output Regulation}
\label{suanfa}
\begin{algorithmic}[1]
        \STATE  Collect $x_{i}(t)$, $u_{i}(t)$, $e_{i}(t)$, $v_{i}(t)$.
        \STATE Express $x_{i}(t)$, $u_{i}(t)$, $e_{i}(t)$, $v_{i}(t)$ using OPBs and collect its coefficients into matrices $X_{i},U_{i},\Im_{i},V_{i}$;
        \STATE Compute $\mu$ through (\ref{d31});
        \STATE Calculate $K_{1i}$ according to \eqref{d28} and \eqref{d17.7};
        \STATE Calculate $M_{i}$ based on (\ref{d33});
        \STATE Compute $K_{2i}$  based on (\ref{d34});
        \STATE Substitute $K_{1i}$, $K_{2i}$, and $\mu$ into (\ref{d3});
        \STATE Return $u_{i}(t)$;
\end{algorithmic}
\end{algorithm}

When $d(t)=0$, the cooperative output regulation problem becomes the output synchronization problem for the multi-agent systems. The set (\ref{d27})  becomes 
\begin{equation}
\label{d39}
\begin{aligned}
{\sum}_{F1i}^{e}=\left\{\begin{bmatrix} A_{i} & B_{i} \\ C_{i} & D_{i} \end{bmatrix}\Bigg|\begin{bmatrix} A_{i} & B_{i} \\ C_{i} & D_{i} \end{bmatrix}\begin{bmatrix} X_{i}\\ U_{i} \end{bmatrix}=\begin{bmatrix} X_{i}\mathcal{D}_{11}\\ Y_{i} \end{bmatrix}\right\},
\end{aligned}
\end{equation}
where $Y_{i}=\begin{bmatrix} \tilde{y}_{i0} & \tilde{y}_{i1} & \tilde{y}_{i2} & ... & \tilde{y}_{iN} \end{bmatrix}$ represents the output matrix of the $i$-th follower and $\tilde{y}_{ik}$ is the coefficient vector corresponding to $b_{k}$. The leader system set can be expressed as
\begin{equation}
\label{d41}
\begin{aligned}
{\sum}_{L1}^{e}=\left\{\begin{bmatrix} A_{0}  \\ C_{0}\end{bmatrix}\Bigg|\begin{bmatrix} A_{0} \\ C_{0} \end{bmatrix}X_{0}=\begin{bmatrix} X_{0}\mathcal{D}_{11}\\ Y_{0} \end{bmatrix}\right\},
\end{aligned}
\end{equation}
where $X_{0}=\begin{bmatrix} \tilde{x}_{00} & \tilde{x}_{01} & \tilde{x}_{02} & ... & \tilde{x}_{0N} \end{bmatrix}$ represents the state matrix of the leader, and $Y_{0}=\begin{bmatrix} \tilde{y}_{00} & \tilde{y}_{01} & \tilde{y}_{02} & ... & \tilde{y}_{0N} \end{bmatrix}$ represents the output matrix of the leader. Here, $\tilde{y}_{0k}$ and $\tilde{x}_{0k}$ are the coefficient vectors corresponding to $b_{k}$. The following corollary addresses the output synchronization for the multi-agent systems when the leader system (\ref{d2.1}) is unknown.

\begin{corollary} \label{co3}
Under Assumptions \ref{a3}-\ref{a4}, suppose the leader system (\ref{d2.1}) is unknown. Assume that all elements of matrix $A_0$ satisfy $|A_{0_{ij}}| \leq \epsilon_{1}$, $\mu$ is sufficiently large, and $X_{0}$ has full row-rank. The data $(X_{i}, U_{i},\Im_{i})$ are informative for output synchronization for all $i \in \{1,...,N\}$, if and only if, the following two conditions are satisfied:
\begin{itemize}
\item The data-driven regulator equations have solutions $M_{i}$:
\begin{equation}
\label{d33.11}
\begin{aligned}
X_{i}\mathcal{D}_{11}M_{i}&=X_{i}M_{i}X_{0}\mathcal{D}_{11}X_{0}^{+}, \\
Y_{i}M_{i}&=Y_{0}X_{0}^{+}.
\end{aligned}
\end{equation}
  \item The data $(X_i, U_i)$ are informative for stabilizability.
  \end{itemize}
\end{corollary}

\begin{proof}
%Sufficiency:
Since $X_{0}$ has full row-rank, Equation (\ref{d41}) becomes 
\begin{equation}
\label{d42}
\begin{aligned}
\begin{bmatrix} A_{0} \\ C_{0} \end{bmatrix}=\begin{bmatrix} X_{0}\mathcal{D}_{11}\\ Y_{0} \end{bmatrix}X_{0}^{+}.
\end{aligned}
\end{equation}
Substituting Equation (\ref{d42}) into Equation (\ref{d33.11}), the following equations can be obtained:
\begin{equation}
\label{d43}
\begin{aligned}
X_{i}\mathcal{D}_{11}M_{i}&=X_{i}M_{i}A_{0}, \\
Y_{i}M_{i}&=C_{0}.
\end{aligned}
\end{equation} 
    Combining Equations (\ref{d39}) and (\ref{d43}), we have
\begin{equation}
\label{d44}
\begin{aligned}
 \Pi_{i} A_{0}&=A_{i} \Pi_{i}+B_{i} \Gamma_{i},\\
  C_{0}&=C_{i}\Pi_{i}+D_{i} \Gamma_{i}.\\
\end{aligned}
\end{equation} 
The rest of the proof is similar to Theorem \ref{t4}, which is omitted here.
%Necessity: The proof is similar to the necessity of Theorem \ref{t4}, hence it is omitted.
\end{proof}

\subsection{Noisy data case}
This subsection examines the cooperative output regulation problem using noisy data, i.e., $W_{i}\neq 0$. 

According to Equation (\ref{d10}), $W_{i}$ can be computed by
\begin{equation}
\label{d49}
\begin{aligned}
W_{i}&=X_{i}\mathcal{D}_{11}-A_{i}X_{i}-B_{i}U_{i}-E_{i}V\\
&=\begin{bmatrix} I_{i} & A_{i} & B_{i}  \end{bmatrix}\begin{bmatrix} X_{i}\mathcal{D}_{11}-E_{i}V \\ -X_{i}\\ -U_{i} \end{bmatrix}.
\end{aligned}
\end{equation}
Substituting Equation (\ref{d49}) into Equation (\ref{d8.8}) yields 
\begin{equation}
\label{d48}
\begin{aligned}
\begin{bmatrix} I_{i} & A_{i} & B_{i}  \end{bmatrix} N_{i} \begin{bmatrix} I_{i} \\ A_{i}^{T} \\ B_{i}^{T} \end{bmatrix} \geq 0,
\end{aligned}
\end{equation}
where
\begin{equation}
\label{d50}
\begin{aligned}
N_{i}\!=\!& \left[ \begin{array} {cc}
I_{i} & X_{i}\mathcal{D}_{11}-E_{i}V \\ \hdashline
0 & -X_{i}\\
0 & -U_{i} \\
\end{array} \right]
\begin{bmatrix} cI_{i} & 0 \\ 0 & -I_{i}\end{bmatrix}\\
&\left[ \begin{array} {cc}
I_{i} & X_{i}\mathcal{D}_{11}-E_{i}V \\ \hdashline
0 & -X_{i}\\
0 & -U_{i} \\
\end{array} \right]^{T}
= \begin{bmatrix} N_{11i} & N_{12i} \\ N_{12i}^{T} & N_{22i}  \end{bmatrix}.\\
\end{aligned}
\end{equation}

The following definition is proposed for the consequent discussions.

\begin{definition} 
The data $(X_{i}, U_{i})$ are informative for quadratic stabilization, if there exist matrices $K_{1i}$ and $Q_{i}=Q^{T}_{i}>0$, such that the Lyapunov function $V_{i}(t)=x_{i}^{T}(t)Q_{i}x_{i}(t)$ satisfies $(A_{i}+B_{i}K_{1i})^{T}Q_{i}+Q_{i}(A_{i}+B_{i}K_{1i})<0$ for all $(A_{i},B_{i})\in$ (\ref{d9}a), where $i \in \{1,...,N\}$. 

\end{definition}

Let $P_{i}=Q_{i}^{-1}>0$. The following inequality can be deduced
\begin{equation}
\label{d52}
\begin{aligned}
(A_{i}+B_{i}K_{1i})P_{i}+P_{i}(A_{i}+B_{i}K_{1i})^{T}<0,
\end{aligned}
\end{equation}
Inequality (\ref{d52}) can be expressed as
$$
\begin{bmatrix} I_{i} & A_{i} & B_{i}  \end{bmatrix} L_{i} \begin{bmatrix} I_{i} \\ A_{i}^{T} \\ B_{i}^{T}  \end{bmatrix} >0,
$$
where 
$$
L_{i}=  
\left[\begin{array}{c:cc}
0 &  -P_{i} & -P_{i}K_{1i} \\ \hdashline
-P_{i} & 0 & 0 \\ -K_{1i}^{T}P_{i} & 0 & 0  
\end{array}\right]
=\begin{bmatrix}
L_{11i} & L_{12i}\\ 
L_{12i}^{T} & L_{22i}
\end{bmatrix}.
$$

The following lemma is referenced to address the quadratic stabilization when $W_{i} \neq 0$.

\begin{lemma} (see \cite{van2020noisy11})\label{l8} Let $ L=\begin{bmatrix} L_{11} & L_{12} \\ L^{T}_{12} & L_{22} \end{bmatrix}, N=\begin{bmatrix}  N_{11} & N_{12} \\ N^{T}_{12} & N_{22} \end{bmatrix} $. Assume $L_{22}\leq0$, $N_{22}\leq0$, and $ker N_{22} \subseteq ker N_{12} $. Suppose that there exist some matrices $\bar{G}$ satisfying the generalized Slater condition
$\begin{bmatrix} I \\ \bar{G} \end{bmatrix}^{T}N\begin{bmatrix} I \\ \bar{G} \end{bmatrix}\geq 0.$
Then, $\begin{bmatrix} I \\ G \end{bmatrix}^{T}L\begin{bmatrix} I \\ G \end{bmatrix}>0$ for all $\begin{bmatrix} I \\ G \end{bmatrix}^{T}N\begin{bmatrix} I \\ G \end{bmatrix}\geq 0$, if and only if, there exist $\alpha \geq 0$ and $\beta>0$ such that
\begin{equation} 
\label{d55}  L-\alpha N \geq \begin{bmatrix} \beta I & 0\\ 0 & 0 \end{bmatrix}. 
\end{equation}
\end{lemma}

Since $N_{22i}=-\begin{bmatrix} X_{i} \\ U_{i} \end{bmatrix} \begin{bmatrix} X_{i}^{T} & U_{i}^{T} \end{bmatrix}<0$, $N_{12i}=(X_{i}\mathcal{D}_{11}-E_{i}V)\begin{bmatrix} X_{i}^{T} & U_{i}^{T} \end{bmatrix}$, it is obvious that $\ker N_{22i} \subseteq  \ker N_{12i}$. Besides, $L_{22i}=\begin{bmatrix} 0&0\\0&0 \end{bmatrix}$. Therefore, $L_{i}$ and $N_{i}$ satisfy Lemma \ref{l8}.

 According to Lemma \ref{l8}, the following theorem provides the necessary and sufficient condition such that the data $(X_{i}, U_{i})$ are informative for quadratic stabilization. Additionally, it also offers a method to compute  $K_{1i}$. 

\begin{theorem}\label{t5}
Let $G_{i}=\begin{bmatrix} A_{i} \\ B_{i}\end{bmatrix}$, $\bar{G}_{i}=\begin{bmatrix} \bar{A}_{i} \\ \bar{B}_{i}\end{bmatrix}$, and suppose $\begin{bmatrix} I \\ \bar{G}_{i} \end{bmatrix}^{T} N_{i}\begin{bmatrix} I \\ \bar{G}_{i} \end{bmatrix}>0$ holds. Then, the data $(X_{i}, U_{i})$ are informative for quadratic stabilization, if and only if, there exist $\alpha_{i} \geq 0$, $\beta_{i}>0$, $P_{i}$ and $\mathcal{J}_{i}$, such that the following inequality holds:
\begin{equation}
\label{d56}
\begin{aligned}
&\left[\begin{array}{c:cc}
-\beta_{i} &  -P_{i} & -\mathcal{J}_{i}^{T} \\ \hdashline
-P_{i} & 0 & 0 \\ -\mathcal{J}_{i} & 0 & 0  
\end{array}\right]
-\alpha_{i} N_{i}
%\begin{bmatrix} I_{i} & X_{i}\mathcal{D}_{11}-E_{i}V \\ 0 & -X_{i}\\ 0 & -U_{i}  \end{bmatrix} \\  
%&\begin{bmatrix} cI_{i} & 0 \\ 0 & -I_{i} \end{bmatrix}
%\begin{bmatrix} I_{i} & X_{i}\mathcal{D}_{11}-E_{i}V \\ 0 & -X_{i}\\ 0 & -U_{i}  \end{bmatrix}^{T}
\geq 0,
\end{aligned}
\end{equation}
and $i \in \{1,...,N\}$. Besides, if $P_{i}$ and $\mathcal{J}_{i}$ satisfy Inequality~\eqref{d56}, $K_{1i}$ is determined by
\begin{equation} \label{d4525}
    K_{1i}=\mathcal{J}_{i}P^{-1}_{i}. 
\end{equation}
\end{theorem}

%\begin{proof} 
%Sufficiency:  the following inequality can be obtained 
%\begin{equation} 
%\label{d59}  L_{i}-\alpha_{i} N_{i} \geq \begin{bmatrix} \beta_{i} I_{i} & 0\\ 0 & 0 \end{bmatrix}. 
%\end{equation}
%Then, the data are informative for quadratic stabilization based on Lemma \ref{l8}.

%Necessity: Equation (\ref{d59}) holds since the data are informative for quadratic stabilization.
%\end{proof}

Due to the noisy data, Equation (\ref{d34.1}) cannot obtain exact solutions. Corresponding approximate regulator equations are formulated as follows
\begin{equation}
\label{d60}
\begin{aligned}
 \omega_{i}+\Pi_{i} S&=A_{i} \Pi_{i}+B_{i} \Gamma_{i}+E_{i},\\
  0&=C_{i}\Pi_{i}+D_{i} \Gamma_{i}+F_{i},\\
\end{aligned}
\end{equation}
where $\omega_{i}$ is the unknown error matrix caused by the noisy data and $i \in \{1,...,N\}$. 

The following definition characterizes the informativity for the approximate cooperative output regulation problem.
\begin{definition}
Under Assumptions \ref{a2}-\ref{a4}, assume that the eigenvalues of $S$ lie on the imaginary axis, and that $\mu$ is sufficiently large. There exists a controller \eqref{d3} such that the data $(X_{i}, U_{i},\Im_{i})$ are informative for the approximate cooperative output regulation problem, if the data $(X_{i}, U_{i})$ are informative for the stabilizability of $A_{C}$ and
$\sum_{Fi} \subseteq \sum_{appr}$, where $\sum_{appr}=\left\{\begin{bmatrix}
    A_{i}&B_{i}\\C_{i}&D_{i}
\end{bmatrix} \bigg| \textit{Equation \eqref{d60} has a solution $(\Pi_{i},\Gamma_{i})$}  \right \}$ and $i \in \{1,...,N\}$. 
\end{definition}

The unknown approximate regulator equations \eqref{d60} can be determined by solving the following optimization problem:
\begin{equation}
\label{d61}
\begin{aligned}
\min_{\Pi_{i}, \Gamma_{i}}\| \omega_{i} \|=\min_{\Pi_{i}, \Gamma_{i}}&\| A_{i} \Pi_{i}+B_{i} \Gamma_{i}+E_{i}-\Pi_{i} S\|, \\
s.t.  \quad -F_{i}&=(\Im_{i}-F_{i}V)M_{i},
\end{aligned}
\end{equation}
where the constraint is derived from Equation \eqref{d33}, which is $0=C_{i}\Pi_{i}+D_{i} \Gamma_{i}+F_{i}$. Based on Equation \eqref{d61}, the following proposition demonstrates that the norm  $\|\omega_{i}\|$ is bounded.

\begin{proposition}
Under Assumptions \ref{a3}-\ref{a4}, assume that the eigenvalues of $S$ lie on the imaginary axis, $\mu$ is sufficiently large, and $col(X_{i},U_{i})$ has full row-rank. Then, the norm of the error matrix $\| \omega_{i} \|$ is bounded.
\end{proposition}

\begin{proof}
Since $\mathrm{col}(X_{i},U_{i})$ has full row-rank, the following equation holds:
$$
\begin{bmatrix}A_{i}&B_{i}\\C_{i}&D_{i}\end{bmatrix}=\begin{bmatrix} X_{i}\mathcal{D}_{11}-E_{i}V-W_{i}\\ \Im_{i}-F_{i}V \end{bmatrix}\begin{bmatrix} X_{i}\\ U_{i} \end{bmatrix}^{+}.
$$
Combing Equations \eqref{d22.2} and \eqref{d61}, the following inequality is obtained:
\begin{equation}
\begin{aligned}
\label{d655}
 &\| \omega_{i} \|\\
=&\left\| (X_{i}\mathcal{D}_{11}-E_{i}V-W_{i}) \begin{bmatrix}X_{i}\\U_{i}\end{bmatrix}^{+}\begin{bmatrix}X_{i}\\U_{i}\end{bmatrix}M_{i}+E_{i}-X_{i}M_{i} S\right\|\\
\leq&\left\| (X_{i}\mathcal{D}_{11}-E_{i}V) \begin{bmatrix}X_{i}\\U_{i}\end{bmatrix}^{+}
\begin{bmatrix}X_{i}\\U_{i}\end{bmatrix}M_{i}+E_{i}-X_{i}M_{i} S\right\|\\
&+\left\|W_{i}\begin{bmatrix}X_{i}\\U_{i}\end{bmatrix}^{+}\begin{bmatrix}X_{i}\\U_{i}\end{bmatrix}M_{i}\right\|\\
\leq&\left\| (X_{i}\mathcal{D}_{11}-E_{i}V) \begin{bmatrix}X_{i}\\U_{i}\end{bmatrix}^{+}
\begin{bmatrix}X_{i}\\U_{i}\end{bmatrix}M_{i}+E_{i}-X_{i}M_{i} S\right\|\\
&+\sqrt{c_{i}}\left\|\begin{bmatrix}X_{i}\\U_{i}\end{bmatrix}^{+}\begin{bmatrix}X_{i}\\U_{i}\end{bmatrix}M_{i}\right\|,
\end{aligned}
\end{equation}
where the final inequality uses Inequality \eqref{d8.8}. Therefore, the norm $\| \omega_{i} \|$ is bounded.
\end{proof}

The norm $\| \omega_{i} \|$ decreases as the bound of the noise $\|W_{i}\|$ decreases. The following theorem demonstrates that the tracking error~$e(t)$ is ultimately uniformly bounded. Besides, it provides a method to compute $K_{2i}$.

\begin{theorem}
\label{t6}
Under Assumptions \ref{a3}-\ref{a4}, assume that the eigenvalues of $S$ lie on the imaginary axis, $\mu$ is sufficiently large, and $col(X_{i},U_{i})$ has full row-rank. The tracking error $e_{i}(t)$ is ultimately uniformly bounded for all $i \in \{1,...,N\}$. Moreover, $K_{2i}$ is given by  $  K_{2i}=(U_{i}-K_{1i}X_{i})M_{i}$, where $M_{i}$ can be computed as follows: 
   \begin{equation} \label{d80}
   \begin{aligned}
  \min_{M_{i}}&\Bigg\{\left\| (X_{i}\mathcal{D}_{11}-E_{i}V) \begin{bmatrix}X_{i}\\U_{i}\end{bmatrix}^{+}
\begin{bmatrix}X_{i}\\U_{i}\end{bmatrix}M_{i}+E_{i}-X_{i}M_{i} S\right\|\\
&+\sqrt{c_{i}}\left\|\begin{bmatrix}X_{i}\\U_{i}\end{bmatrix}^{+}\begin{bmatrix}X_{i}\\U_{i}\end{bmatrix}M_{i}\right\| \Bigg\},
 \\
&s.t.  \quad -F_{i}=(\Im_{i}-F_{i}V)M_{i}.
\end{aligned}
\end{equation} 
\end{theorem}

\begin{proof} 
 Due to the presence of noise $W_{i}$, Equation \eqref{d3.3} is modified to
\begin{equation}
\label{d63}
\begin{aligned}
X_{c}\mathcal{D}_{11}=&\begin{bmatrix} A+BK_{1} & BK_{2}\\ 0 & (I_{N}\otimes S)-\mu (H\otimes I_{q}) \end{bmatrix}X_{c} \\ 
&+\begin{bmatrix} E \\ \mu(H\otimes I_{q}) \end{bmatrix} \bar{V}+\begin{bmatrix} W \\ \tilde{W} \end{bmatrix},\\
\Im=&\begin{bmatrix} C+DK_{1} & DK_{2}  \end{bmatrix}X_{c}+F\bar{V},
\end{aligned}
\end{equation}
where $W=\mathrm{col}(W_{1},\cdots,W_{N})$ and $\tilde{W}=\mathrm{col}(\tilde{W}_{1}, \cdots,\tilde{W}_{N})$, with $\tilde{W}_{i}$ representing the noisy data about $\tilde{\eta}_{i}$. Let $\bar W=\begin{bmatrix} W \\ \tilde{W} \end{bmatrix}$, Equation (\ref{d63}) can be expressed as 
\begin{equation}
\label{d64}
 \begin{aligned}
X_{c}\mathcal{D}_{11}&=A_{C}X_{c}+B_{C}\bar{V}+\bar W,\\
\Im&=C_{C}X_{c}+D_{C}\bar{V}.
\end{aligned}
\end{equation} 
Let $\Pi_{C}=\begin{bmatrix} \Pi \\ I_{qN} \end{bmatrix}$, $\omega=\mathrm{blockdiag}(\omega_{1},\omega_{2},\cdots,\omega_{N})$, $\Pi=\mathrm{blockdiag}(\Pi_{1},\Pi_{2},\cdots,\Pi_{N})$ and $\omega_{C}=\begin{bmatrix} \omega \\ 0 \end{bmatrix}$. The compact form of Equation (\ref{d60}) is
 \begin{align*}
\omega_{C}+\Pi_{C}( I \otimes S)&=A_{C} \Pi_{C}+B_{C},\\
  0&=C_{C}\Pi_{C}+D_{C}.\\
\end{align*}
Let $\bar X=X_{C}-\Pi_{C}\bar V$. Equation (\ref{d64}) can be expressed as
\begin{align*}
\bar X\mathcal{D}_{11}&=A_{C}\bar X+(A_{C}\Pi_{C}+B_{C})\bar V+\bar W\\
&=A_{C}\bar X+[A_{C}\Pi_{C}+B_{C}-\Pi_{C}(I_{N}\otimes S)]\bar V+\bar W\\
&=A_{C}\bar X+(\omega_{C} \bar V+\bar W),\\
\Im&=C_{C}X_{C}+D_{C}\bar V
=C_{C}X_{C}+(C_{C}\Pi_{C}+D_{C})\bar V\\
&=C_{C}\bar X.\\
\end{align*}
Clearly, $\omega_C$ and $\bar{W}$ are bounded. Additionally,  the eigenvalues of $S$ lie on the imaginary axis. Hence, $(\omega_{C} \bar V+\bar W)\mathrm{b}$ is bounded, implying that Inequality (\ref{d25}) holds. Consequently, the tracking error $e(t)$ is ultimately uniformly bounded. By invoking Equation \eqref{d22.2}, one obtains $K_{2i}=(U_{i}-K_{1i}X_{i})M_{i}$. The optimization problem \eqref{d80} is derived from \eqref{d655}. 
\end{proof}

The steps for computing the controller \eqref{d3} are similar to Algorithm \ref{suanfa}, which are omitted here.

\section{NUMERICAL SIMULATION}
Consider a directed weighted graph composed of $N = 4$
agents and a leader labeled as $0$. Assuming $4 \geq L_{ij}\geq 2$, the graph $\mathcal{G}_{1}$ is depicted in Fig. \ref{fig1}.

\begin{figure}[htb]
  \centering
  \includegraphics[scale=0.4]{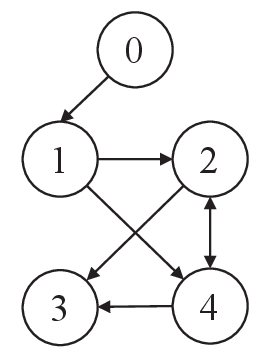}
  \caption{Network topology for the graph $\mathcal{G}_{1}$}
  \label{fig1}
\end{figure}

\begin{figure}[ht]
  \centering
  \includegraphics[width=200pt]{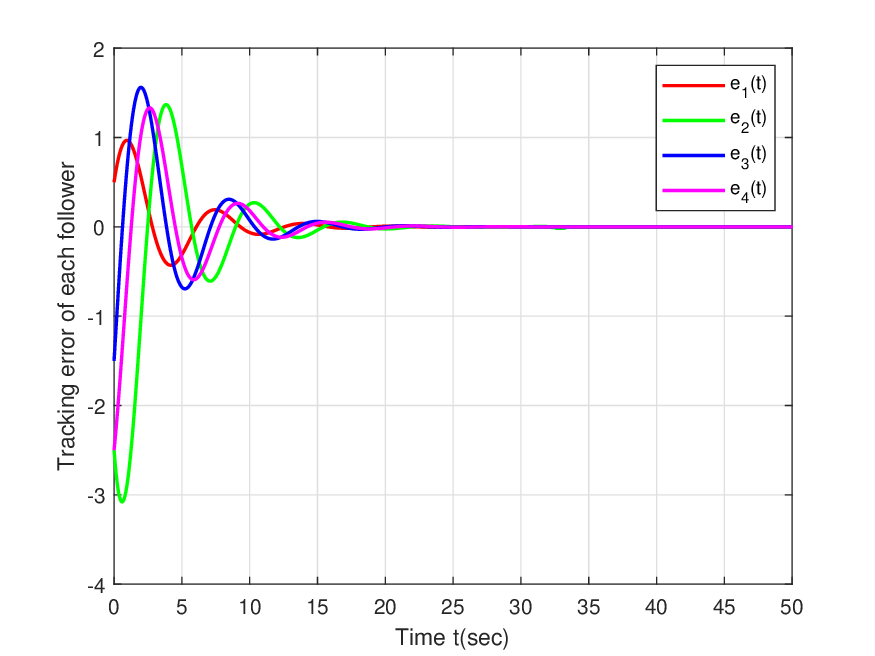}
  \caption{Tracking errors of the followers using exact data.}
  \label{fig2}
\end{figure}
The Laplacian matrix $\mathcal{L}_{1}$ is
$$\mathcal{L}_{1}=\begin{bmatrix}
    0& 0 &0 &0 &0 \\
    -2& 2& 0& 0& 0\\
    0 &-2 &4& 0 &-2\\
    0& 0& -2& 4& -2\\
   0 &-2 &-2& 0& 4\\
\end{bmatrix}.$$ It is easy to obtain $\varepsilon_{1}=2$ and 
$\varepsilon_{2}=4$. Invoking Inequality (\ref{d31}), the bound for $\mu > 128$ is derived.

The leader is given by Equation (\ref{d2}) with $S=\begin{bmatrix} 1 &0\\ 0& 1 \end{bmatrix}$. The followers are given by Equation (\ref{d1}) with
\begin{equation}
\begin{aligned} \label{d212}
&A_{1}=A_{2}=\begin{bmatrix} 0& 1 \\0 &0 \end{bmatrix}, A_{3}=A_{4}=\begin{bmatrix} 0& 1 \\1 &0 \end{bmatrix},\\
&B_{1}=B_{2}=\begin{bmatrix} 0 \\1 \end{bmatrix},
B_{3}=B_{4}=\begin{bmatrix} 1 \\0 \end{bmatrix},\\
&C_{1}=C_{2}=\begin{bmatrix} 1 &0 \end{bmatrix}, 
C_{3}=C_{4}=\begin{bmatrix} 0 &1 \end{bmatrix}, \\
&D_{1}=D_{2}=D_{3}=D_{4}=0, \\
&E_{1}=E_{2}=\begin{bmatrix} 1& 0 \\0 &1 \end{bmatrix}, E_{3}=E_{4}=\begin{bmatrix} 0& 0 \\0 &0 \end{bmatrix},\\
&F_{1}=F_{2}=F_{3}=F_{4}=\begin{bmatrix} -1& 0 \end{bmatrix}.
\end{aligned}
\end{equation}
Let $x_{i}(0)=\begin{bmatrix} 1 & 1\end{bmatrix}^{T}$, $v(0)= \begin{bmatrix} 0.5 & 0.5\end{bmatrix}^{T}$ and $u_{i}(t)=e^{-t}$. Then, 
the state trajectories are as follows:
\begin{align*} 
&x_{1}(t)=x_{2}(t)= \begin{bmatrix} \frac{3}{2}t+2cosh(t)-1\\ \frac{1}{2}e^{t}-e^{-t}+\frac{3}{2}\end{bmatrix},\\
&x_{3}(t)=x_{4}(t)= \begin{bmatrix} e^{t} +\frac{1}{4}e^{-t}(2t + e^{2t} - 1)\\ e^{t}-\frac{1}{4}e^{-t}(2t + e^{2t} - 1)\end{bmatrix},\\
&v(t)= \begin{bmatrix} \frac{1}{2}e^{t}\\\frac{1}{2}e^{t}\end{bmatrix}.
\end{align*}
Using Chebyshev polynomials, $\mathcal{D}$ is computed as
$$\mathcal{D}=\begin{bmatrix}
0 &0& 0 &0 &0 &...\\1& 0& 0 &0& 0& ...\\0& 4& 0 &0& 0& ...\\3& 0& 6 &0& 0& ...\\0& 8& 0 &8& 0& ...\\5& 0& 10 &0& 10& ...\\ \vdots& \vdots& \vdots &\vdots&\vdots& \ddots
\end{bmatrix}.$$ 
Using the Chebfun toolbox \cite{driscoll2014chebfun}, it can be verified that $16$ samples are sufficient to compute the Chebyshev representation of the signals with machine precision, refer to Table \ref{tab5}.
\begin{table}
\begin{center}
\caption{CHEBYSHEV COEFFICIENTS}
\label{tab5}
\begin{tabular}{| c | c | c | }
\hline
$U^{T}$ & $X_{1}^{T}$ & $V^{T}$  \\
\hline
$1.26$ & $1.53$ $8.67$·$10^{-1}$& $6.33$·$10^{-1}$ $6.33$·$10^{-1}$ \\
 $-1.13$& $1.50$ $1.70$ & $5.65$·$10^{-1}$ $5.65$·$10^{-1}$ \\
  $2.72$·$10^{-1}$&$5.43$·$10^{-1}$ $-1.36$·$10^{-1}$& $1.36$·$10^{-1}$ $1.36$·$10^{-1}$ \\
$-4.43$·$10^{-2}$& $1.11$·$10^{-16}$ $6.65$·$10^{-2}$& $2.20$·$10^{-2}$ $2.20$·$10^{-2}$ \\
 $5.5$·$10^{-3}$& $1.11$·$10^{-2}$ $2.70$·$10^{-3}$& $2.74$·$10^{-3}$ $2.74$·$10^{-3}$ \\
 $5.43$·$10^{-4}$& $2.29$·$10^{-17}$ $8.14$·$10^{-4}$& $2.72$·$10^{-4}$ $2.72$·$10^{-4}$ \\
  $4.50$·$10^{-5}$&$9.00$·$10^{-5}$ $-2.25$·$10^{-5}$& $2.25$·$10^{-5}$ $2.25$·$10^{-5}$  \\
  $-3.20$·$10^{-6}$&$4.65$·$10^{-18}$ $4.80$·$10^{-6}$& $1.60$·$10^{-6}$ $1.60$·$10^{-6}$  \\
 $1.99$·$10^{-7}$& $3.98$·$10^{-7}$ $-9.96$·$10^{-8}$&$9.96$·$10^{-8}$ $9.96$·$10^{-8}$ \\
  $-1.10$·$10^{-8}$& $-2.78$·$10^{-17}$ $1.66$·$10^{-8}$& $5.52$·$10^{-9}$ $5.52$·$10^{-9}$\\
  $5.51$·$10^{-10}$& $1.10$·$10^{-9}$ $-2.75$·$10^{-10}$&$2.75$·$10^{-10}$ $2.75$·$10^{-10}$ \\
 $2.50$·$10^{-11}$& $3.66$·$10^{-17}$ $3.75$·$10^{-11}$& $1.25$·$10^{-11}$ $1.25$·$10^{-11}$ \\
  $1.04$·$10^{-12}$& $2.08$·$10^{-12}$ $-5.20$·$10^{-13}$& $5.20$·$10^{-13}$ $5.20$·$10^{-13}$\\
  $-4.00$·$10^{-14}$& $2.00$·$10^{-17}$ $5.99$·$10^{-14}$& $2.00$·$10^{-14}$ $2.00$·$10^{-14}$\\
  $1.43$·$10^{-15}$& $2.83$·$10^{-15}$ $-6.45$·$10^{-16}$& $7.15$·$10^{-16}$ $7.15$·$10^{-16}$ \\
  $0.00$& $0.00$ $0.00$& $0.00$ $0.00$ \\
\hline
 $\Im_{1}^{T}$ & $\Im_{3}^{T}$ & $X_{3}^{T}$ \\
 \hline
 $9.00$·$10^{-1}$ & $3.51$·$10^{-1}$ & $9.83$·$10^{-1}$ $1.55$ \\
 $9.35$·$10^{-1}$ & $1.83$ & $2.39$ $9.95$·$10^{-1}$ \\
$4.07$·$10^{-1}$ &  $-1.58$·$10^{-1}$ & $-2.22$·$10^{-2}$ $5.65$·$10^{-1}$  \\
$-2.22$·$10^{-2}$ & $1.14$·$10^{-1}$ & $1.36$·$10^{-1}$ $-2.74$·$10^{-3}$  \\
$-8.20$·$10^{-3}$ & $-8.50$·$10^{-3}$ & $-5.75$·$10^{-3}$ $1.67$·$10^{-2}$  \\
$-2.71$·$10^{-4}$ & $1.90$·$10^{-3}$ & $2.20$·$10^{-3}$ $-5.65$·$10^{-4}$  \\
$6.74$·$10^{-5}$ & $-1.14$·$10^{-4}$ & $-9.16$·$10^{-5}$ $1.82$·$10^{-4}$ \\
$-1.60$·$10^{-6}$ &  $1.45$·$10^{-5}$ & $1.61$·$10^{-5}$ $-6.50$·$10^{-6}$  \\
$2.99$·$10^{-7}$ &  $-7.03$·$10^{-7}$ &$-6.03$·$10^{-7}$ $1.00$·$10^{-6}$   \\
$-5.52$·$10^{-9}$ &  $6.10$·$10^{-8}$ &$6.65$·$10^{-8}$ $-3.34$·$10^{-8}$  \\
$-8.26$·$10^{-10}$ &  $-2.49$·$10^{-9}$ &$-2.21$·$10^{-9}$ $3.32$·$10^{-9}$    \\
$-1.25$·$10^{-11}$ &  $1.63$·$10^{-10}$ &$1.75$·$10^{-10}$ $-1.00$·$10^{-10}$    \\
$1.56$·$10^{-12}$ &  $-5.74$·$10^{-12}$ &$-5.22$·$10^{-12}$ $7.29$·$10^{-12}$ \\
$-2.00$·$10^{-14}$ &  $3.00$·$10^{-13}$ &$3.20$·$10^{-13}$ $-2.00$·$10^{-13}$ \\
$2.12$·$10^{-15}$ &  $-9.25$·$10^{-15}$ &$-8.53$·$10^{-15}$ $1.15$·$10^{-14}$   \\
$0.00$ &  $4.30$·$10^{-16}$ & $4.30$·$10^{-16}$ $-2.84$·$10^{-16}$  \\
\hline
\end{tabular}
\end{center}
\end{table}
Invoking Equation (\ref{d28}), one obtains $K_{11}=K_{21}=\begin{bmatrix}-1.000&-0.500\end{bmatrix}$ and $K_{31}=K_{41}=\begin{bmatrix}-0.500&-2.000\end{bmatrix}$. 
Calculating Equation (\ref{d29}), one obtains
$$ \mathrm{rank}\begin{bmatrix} X_{i}\mathcal{D}_{\mathbin{b}}-E_{i}V-\lambda X_{i} \\ \Im_{i}-F_{i}V \end{bmatrix}=3.$$
It proves that the data  $(X_{i}, U_{i},\Im_{i})$  are informative for data-driven cooperative output regulation problem. By combing Equations (\ref{d33}) and (\ref{d34}), one obtains $K_{12}=K_{22}=\begin{bmatrix}1.000&-1.000\end{bmatrix}$ and $K_{32}=K_{42}=\begin{bmatrix}2.500&0.000\end{bmatrix}$.
%Let $v(0)=\begin{bmatrix} 0.5 & 0.5 \end{bmatrix}^{T}$, $x_{1}(0)=\begin{bmatrix} 1 & 1 \end{bmatrix}^{T}$, $x_{2}(0)=\begin{bmatrix} -2 & -2 \end{bmatrix}^{T}$,
%$x_{3}(0)=\begin{bmatrix} 3 & -1 \end{bmatrix}^{T}$,
%$x_{4}(0)=\begin{bmatrix} 2 & -2 \end{bmatrix}^{T}$, $\eta_{1}(0)=\begin{bmatrix} 0 & 0 \end{bmatrix}^{T}$,
%$\eta_{2}(0)=\begin{bmatrix} 0 & 0 \end{bmatrix}^{T}$,
%$\eta_{3}(0)=\begin{bmatrix} 0 & 0 \end{bmatrix}^{T}$,
%$\eta_{4}(0)=\begin{bmatrix} 0 & 0 \end{bmatrix}^{T}$.
The simulation results are depicted in Fig. \ref{fig2}.  The tracking errors of all the followers converge asymptotically to zero. 

Let the dynamic topology switch from graph $\mathcal{G}_{2}$ to graph $\mathcal{G}_{3}$ at $t=10s$, see Fig. \ref{fig4}. 
%The weight of the connection from $\mathbf{v}_{3}$ to $\mathbf{v}_{1}$ is $3$, while the weights of other connections are $2$.
\begin{figure}[ht]
  \centering
  \includegraphics[width=200pt]{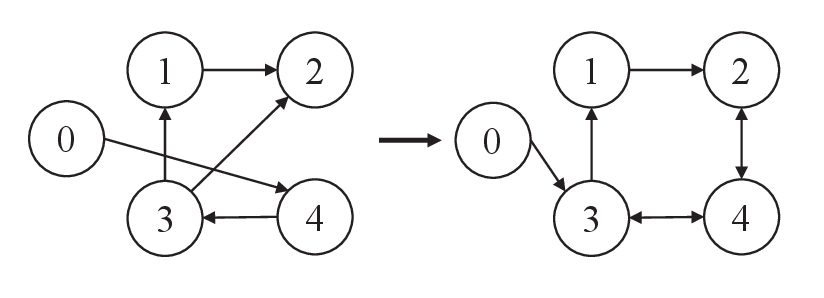}
  \caption{Switching network topology from $\mathcal{G}_{2}$ to $\mathcal{G}_{3}$.}
  \label{fig4}
\end{figure}
The matrices $H_{2}$ and $H_{3}$ are
$$H_{2}=\begin{bmatrix}
     3& 0& -3& 0\\
    -2 &4& -2 &0\\
     0& 0& 2& -2\\
   0 &0& 0& 2\\
\end{bmatrix},H_{3}=\begin{bmatrix}
     3& 0& -3& 0\\
   -2 &4& 0 &-2\\
    0& 0& 4& -2\\
   0 &-2& -2& 4\\
\end{bmatrix}. $$
%Let $x_{1}(0)=\begin{bmatrix} 5 & 5 \end{bmatrix}^{T}$, $x_{2}(0)=\begin{bmatrix} -5 & -5 \end{bmatrix}^{T}$, $x_{3}(0)=\begin{bmatrix} -3 & -3 \end{bmatrix}^{T}$,
%$x_{4}(0)=\begin{bmatrix} 3 & 3 \end{bmatrix}^{T}$ and $\eta_{1}(0)=\eta_{2}(0)=\eta_{3}%(0)=\eta_{4}(0)=\begin{bmatrix} 0 & 0 \end{bmatrix}^{T}$. 
The tracking errors of all the followers also converge asymptotically to zero, see Fig. \ref{fig3}.
\begin{figure}[ht]
  \centering
  \includegraphics[width=200pt]{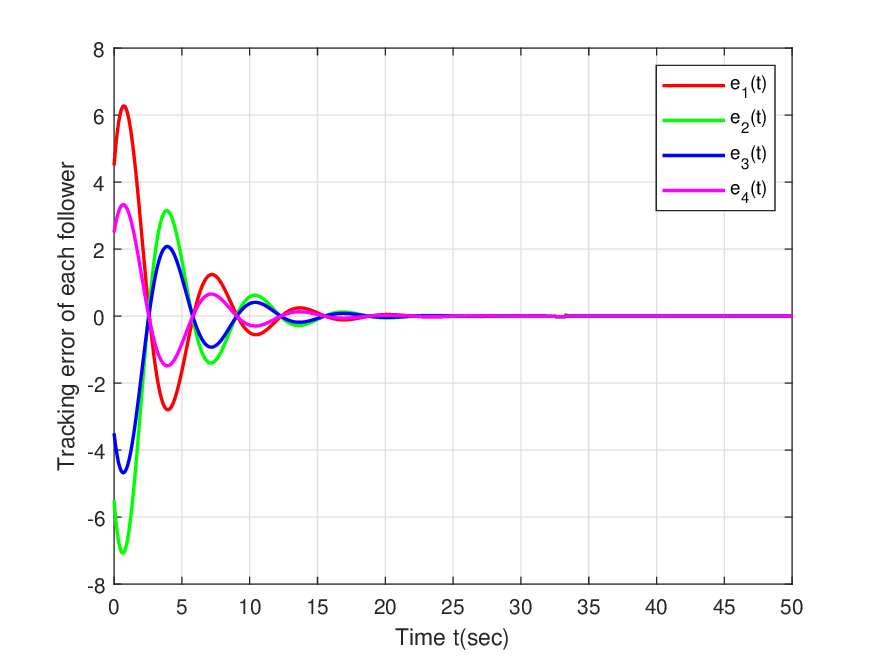}
  \caption{Tracking errors of the followers using exact data.}
  \label{fig3}
\end{figure}

In the noisy data case, the leader is given by Equation (\ref{d2}) with $S=\begin{bmatrix} 0 &1\\ -1& 0 \end{bmatrix}$. The followers are described by Equation (\ref{d1}) with matrices \eqref{d212}. According to Lemma \ref{l3}, one obtains $c_1=0.16$ and $c_2=0.04$. The norms of the error matrices satisfy $\| \omega_{1} \|\leq 1.57$ and $\| \omega_{3} \|\leq 1.19$ according to Inequality \eqref{d655}.  Applying Equations \eqref{d56} and (\ref{d4525}), one obtains $K_{11}=K_{21}=\begin{bmatrix} -2.8176&-3.2005\end{bmatrix}$ and $K_{31}=K_{41}=\begin{bmatrix}-1.0942&-1.0951\end{bmatrix}$. By combing Equations (\ref{d34}) and \eqref{d80}, one obtains $K_{12}=K_{22}=\begin{bmatrix}42.7&8.9\end{bmatrix}$ and $K_{32}=K_{42}=\begin{bmatrix}2.0246&-49.0265\end{bmatrix}$. Figures \ref{fig5} compare the noisy and exact data cases, demonstrating that the controller designed using noisy data ensures that the tracking errors remain bounded. Figure \ref{fig4} demonstrates that the tracking errors are ultimately uniformly bounded.

\begin{figure}[ht]
  \centering
  \includegraphics[width=200pt]{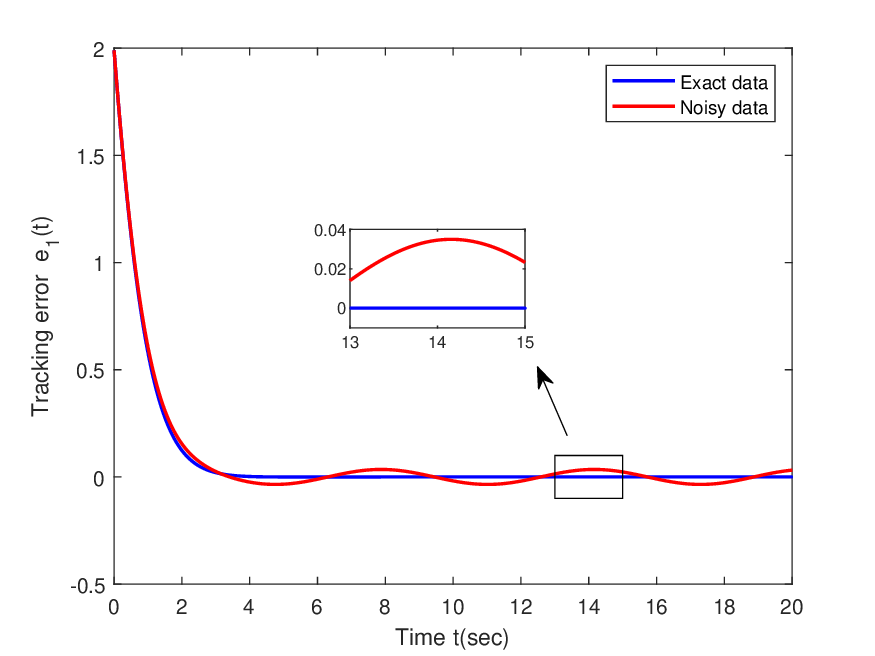}
  \caption{Tracking errors $e_{1}(t)$ using exact data and noisy data.}
  \label{fig5}
\end{figure}

%\begin{figure}[ht]
%  \centering
% \includegraphics[width=200pt]{lianxue2.eps}
 %\caption{Tracking errors $e_{3}(t)$ with exact data and noisy data.}
 % \label{fig6}
%\end{figure}

\begin{figure}[ht]
  \centering
  \includegraphics[width=200pt]{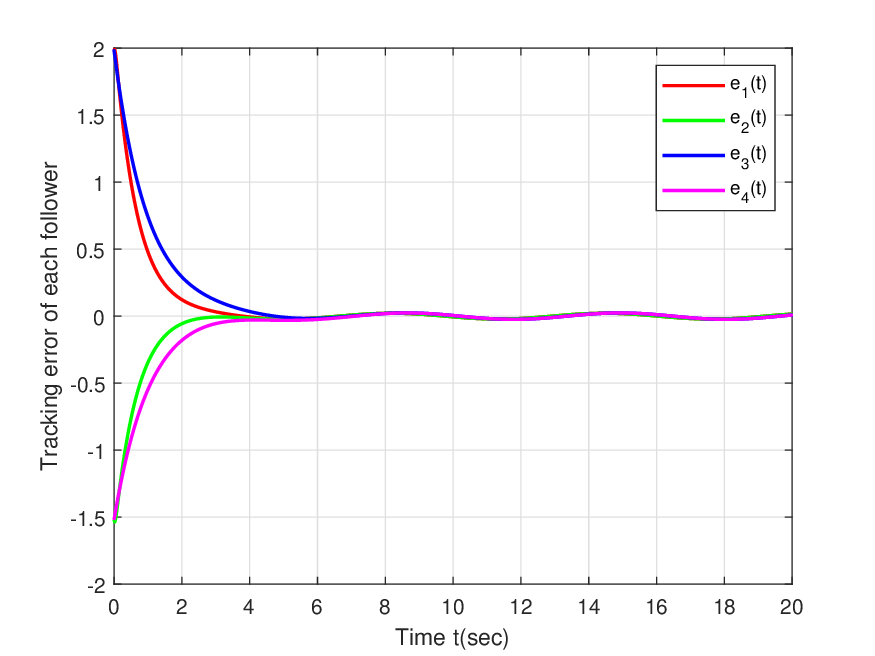}
  \caption{Tracking errors of the followers using noisy data.}
  \label{fig4}
\end{figure}

\section{CONCLUSION}
This paper has studied the data-driven cooperative output regulation problem for continuous-time multi-agent systems with an unknown network topology. For the exact data case, some necessary and sufficient conditions for the data-driven cooperative output regulation problem have been established. %Besides, the distributed controller has been designed based on the exact data and the bound.
In the case of noisy data, the bound on the output error has been found, demonstrating a positive correlation with the noise bound. A lower bound of the minimum non-zero eigenvalue of the Laplacian matrix independent of the network structure has been estimated. Additionally, a distributed controller has been designed. In future studies, the data-driven nonlinear cooperative output regulation problem will be investigated.  

\bibliographystyle{IEEEtran}
\bibliography{IEEEabrv,SMC}

\begin{thebibliography}{10}
\providecommand{\url}[1]{#1}
\csname url@samestyle\endcsname
\providecommand{\newblock}{\relax}
\providecommand{\bibinfo}[2]{#2}
\providecommand{\BIBentrySTDinterwordspacing}{\spaceskip=0pt\relax}
\providecommand{\BIBentryALTinterwordstretchfactor}{4}
\providecommand{\BIBentryALTinterwordspacing}{\spaceskip=\fontdimen2\font plus
\BIBentryALTinterwordstretchfactor\fontdimen3\font minus \fontdimen4\font\relax}
\providecommand{\BIBforeignlanguage}[2]{{%
\expandafter\ifx\csname l@#1\endcsname\relax
\typeout{** WARNING: IEEEtran.bst: No hyphenation pattern has been}%
\typeout{** loaded for the language `#1'. Using the pattern for}%
\typeout{** the default language instead.}%
\else
\language=\csname l@#1\endcsname
\fi
#2}}
\providecommand{\BIBdecl}{\relax}
\BIBdecl

\bibitem{sun2021distributed}
Z.~Sun, A.~Rantzer, Z.~Li, and A.~Robertsson, ``Distributed adaptive stabilization,'' \emph{Automatica}, vol. 129, p. 109616, 2021.

\bibitem{wen2020coordination}
G.~Wen, X.~Yu, W.~Yu, and J.~Lu, ``Coordination and control of complex network systems with switching topologies: A survey,'' \emph{IEEE Transactions on Systems, Man, and Cybernetics: Systems}, vol.~51, no.~10, pp. 6342--6357, 2020.

\bibitem{zhou2021terminal}
J.~Zhou, Y.~Lv, G.~Wen, and G.~Chen, ``Terminal-time synchronization of multivehicle systems under sampled-data communications,'' \emph{IEEE Transactions on Systems, Man, and Cybernetics: Systems}, vol.~52, no.~4, pp. 2625--2636, 2021.

\bibitem{sun2016exponential}
Z.~Sun, S.~Mou, B.~D. Anderson, and M.~Cao, ``Exponential stability for formation control systems with generalized controllers: A unified approach,'' \emph{Systems \& Control Letters}, vol.~93, pp. 50--57, 2016.

\bibitem{xu2023adaptive}
T.~Xu, Z.~Duan, Z.~Sun, and G.~Chen, ``Adaptive distributed formation-containment control on switching directed networks: A dynamic triggering framework,'' \emph{IEEE Transactions on Control of Network Systems}, 2023.

\bibitem{cheng2019cooperative}
B.~Cheng, Z.~Li, and X.~Wang, ``Cooperative output regulation of heterogeneous multi-agent systems with adaptive edge-event-triggered strategies,'' \emph{IEEE Transactions on Circuits and Systems II: Express Briefs}, vol.~67, no.~10, pp. 2199--2203, 2019.

\bibitem{liu2017cooperative}
K.~Liu, Y.~Chen, Z.~Duan, and J.~L{\"u}, ``Cooperative output regulation of lti plant via distributed observers with local measurement,'' \emph{IEEE Transactions on Cybernetics}, vol.~48, no.~7, pp. 2181--2191, 2017.

\bibitem{su2011cooperative}
Y.~Su and J.~Huang, ``Cooperative output regulation of linear multi-agent systems,'' \emph{IEEE Transactions on Automatic Control}, vol.~57, no.~4, pp. 1062--1066, 2011.

\bibitem{su2012cooperative1}
{Y. Su and J. Huang}, ``Cooperative output regulation with application to multi-agent consensus under switching network,'' \emph{IEEE Transactions on Systems, Man, and Cybernetics, Part B (Cybernetics)}, vol.~42, no.~3, pp. 864--875, 2012.

\bibitem{huang2013cooperative1}
C.~Huang and X.~Ye, ``Cooperative output regulation of heterogeneous multi-agent systems: An \(\mathcal{H}_\infty\) criterion,'' \emph{IEEE Transactions on Automatic Control}, vol.~59, no.~1, pp. 267--273, 2013.

\bibitem{huang2016cooperative}
J.~Huang, ``The cooperative output regulation problem of discrete-time linear multi-agent systems by the adaptive distributed observer,'' \emph{IEEE Transactions on Automatic Control}, vol.~62, no.~4, pp. 1979--1984, 2016.

\bibitem{liu2018adaptive}
T.~Liu and J.~Huang, ``Adaptive cooperative output regulation of discrete-time linear multi-agent systems by a distributed feedback control law,'' \emph{IEEE Transactions on Automatic Control}, vol.~63, no.~12, pp. 4383--4390, 2018.

\bibitem{hu2017cooperative}
W.~Hu, L.~Liu, and G.~Feng, ``Cooperative output regulation of linear multi-agent systems by intermittent communication: A unified framework of time-and event-triggering strategies,'' \emph{IEEE Transactions on Automatic Control}, vol.~63, no.~2, pp. 548--555, 2017.

\bibitem{van2012subspace}
P.~Van~Overschee and B.~De~Moor, \emph{Subspace identification for linear systems: Theory—Implementation—Applications}.\hskip 1em plus 0.5em minus 0.4em\relax Springer Science \& Business Media, 2012.

\bibitem{8933093}
C.~De~Persis and P.~Tesi, ``Formulas for data-driven control: Stabilization, optimality, and robustness,'' \emph{IEEE Transactions on Automatic Control}, vol.~65, no.~3, pp. 909--924, 2020.

\bibitem{van2020data}
H.~J. Van~Waarde, J.~Eising, H.~L. Trentelman, and M.~K. Camlibel, ``Data informativity: A new perspective on data-driven analysis and control,'' \emph{IEEE Transactions on Automatic Control}, vol.~65, no.~11, pp. 4753--4768, 2020.

\bibitem{trentelman2021informativity}
H.~L. Trentelman, H.~J. van Waarde, and M.~K. Camlibel, ``An informativity approach to the data-driven algebraic regulator problem,'' \emph{IEEE Transactions on Automatic Control}, vol.~67, no.~11, pp. 6227--6233, 2021.

\bibitem{zhu2024data}
L.~Zhu and Z.~Chen, ``Data informativity for robust output regulation,'' \emph{IEEE Transactions on Automatic Control}, 2024.

\bibitem{liang2024data}
D.~Liang, Y.~Dong, C.~Wang, and G.~Zhai, ``Data-driven cooperative output regulation of linear discrete-time multiagent systems with unknown dynamics,'' \emph{IEEE Transactions on Systems, Man, and Cybernetics: Systems}, 2024.

\bibitem{xie2023data}
K.~Xie, Y.~Jiang, X.~Yu, and W.~Lan, ``Data-driven cooperative optimal output regulation for linear discrete-time multi-agent systems by online distributed adaptive internal model approach,'' \emph{Science China Information Sciences}, vol.~66, no.~7, p. 170202, 2023.

\bibitem{jiao2021data}
J.~Jiao, H.~J. van Waarde, H.~L. Trentelman, M.~K. Camlibel, and S.~Hirche, ``Data-driven output synchronization of heterogeneous leader-follower multi-agent systems,'' in \emph{2021 60th IEEE Conference on Decision and Control (CDC)}.\hskip 1em plus 0.5em minus 0.4em\relax IEEE, 2021, pp. 466--471.

\bibitem{trefethen2019approximation}
L.~N. Trefethen, \emph{Approximation theory and approximation practice, extended edition}.\hskip 1em plus 0.5em minus 0.4em\relax SIAM, 2019.

\bibitem{rapisarda2023orthogonal}
P.~Rapisarda, H.~van Waarde, and M.~Camlibel, ``Orthogonal polynomial bases for data-driven analysis and control of continuous-time systems,'' \emph{IEEE Transactions on Automatic Control}, 2023.

\bibitem{shivakumar1996two}
P.~Shivakumar, J.~J. Williams, Q.~Ye, and C.~A. Marinov, ``On two-sided bounds related to weakly diagonally dominant {M}-matrices with application to digital circuit dynamics,'' \emph{SIAM Journal on Matrix Analysis and Applications}, vol.~17, no.~2, pp. 298--312, 1996.

\bibitem{west2001introduction}
D.~B. West \emph{et~al.}, \emph{Introduction to graph theory}.\hskip 1em plus 0.5em minus 0.4em\relax Prentice hall Upper Saddle River, 2001, vol.~2.

\bibitem{li2016distributed}
X.~Li, M.~Z. Chen, H.~Su, and C.~Li, ``Distributed bounds on the algebraic connectivity of graphs with application to agent networks,'' \emph{IEEE Transactions on Cybernetics}, vol.~47, no.~8, pp. 2121--2131, 2016.

\bibitem{van2020noisy11}
H.~J. van Waarde, M.~K. Camlibel, and M.~Mesbahi, ``From noisy data to feedback controllers: Nonconservative design via a matrix s-lemma,'' \emph{IEEE Transactions on Automatic Control}, vol.~67, no.~1, pp. 162--175, 2020.

\bibitem{driscoll2014chebfun}
T.~A. Driscoll, N.~Hale, and L.~N. Trefethen, ``Chebfun guide,'' 2014.

\end{thebibliography}

\end{document}